%% file: 00main.tex
\documentclass[a4paper,UKenglish,cleveref, autoref, thm-restate]{lipics-v2021}
\usepackage{graphicx} 
\usepackage{lineno}
\usepackage{url}
\usepackage{float}
\input{0macros}

\linenumbers
\title{Lookahead Games and Efficient Determinisation \\ of History-Deterministic B\"uchi Automata}
\titlerunning{Lookahead Games and Efficient Determinisation of HD B\"uchi Automata}
\author{Rohan Acharya}{University of Warwick, Coventry, UK}{rohan.acharya@warwick.ac.uk}{}{Undergraduate Research Support Scheme from the Department of Computer Science at the University of Warwick.
}
\author{Marcin Jurdzi\'nski}{University of Warwick, Coventry, UK \and \url{https://www.dcs.warwick.ac.uk/~mju/}}{marcin.jurdzinski@warwick.ac.uk}{https://orcid.org/0000-0003-3640-8481}{}
\author{Aditya Prakash}{University of Warwick, Coventry, UK \and \url{https://apitya.github.io/}}{aditya.prakash@warwick.ac.uk}{https://orcid.org/0000-0002-2404-0707}{Chancellors' International Scholarship from the University of Warwick.}
\authorrunning{R. Acharya, M. Jurdzi\'nski and A. Prakash} 

\Copyright{R. Acharya, M. Jurdzi\'nski and A. Prakash} 



\category{Track B: Automata, Logic, Semantics, and Theory of Programming} 

\relatedversion{} 

\funding{ We acknowledge the Centre for Discrete Mathematics and Its Applications (DIMAP) at the University of Warwick for partial support.}

\acknowledgements{We thank Denis Kuperberg for several insightful exchanges. 
    We are grateful to the reviewers for their feedback and suggestions on how to improve the paper.}

\nolinenumbers 

\EventEditors{2}
\EventNoEds{2}
\EventLongTitle{Conference of Very Important Things}
\EventShortTitle{CVIT 2020}
\EventAcronym{CVIT}
\EventYear{2016}
\EventDate{December 24--27, 2016}
\EventLocation{Little Whinging, United Kingdom}
\EventLogo{}
\SeriesVolume{42}
\ArticleNo{23}

\begin{document}

\maketitle
\begin{abstract}
    Our main technical contribution is a polynomial-time determinisation procedure for history-deter\-mi\-nis\-tic B\"uchi automata, which settles an open question of Kuperberg and Skrzypczak, 2015. 
    A~key conceptual contribution is the lookahead game, which is a variant of Bagnol and Kuperberg's token game, in which Adam is given a fixed lookahead.
    We prove that the lookahead game is equivalent to the $1$-token game.
    This allows us to show that the $1$-token game characterises history-determinism for semantically-deterministic B\"uchi automata, which paves the way to our polynomial-time determinisation procedure.
\end{abstract}
\section{Introduction}
\input{1intronew}

\section{Preliminaries}
\input{2prelimsnew}

\section{Sufficient to think about Universal Automata}\label{sec:universal-reduction}
\input{3universalreduction}

\section{When 1-Token Game is Enough}
\input{4g1forbuchi}

\input{4jokercounterexample}

\section{Determinising HD B\"uchi Automata in Polynomial Time}
\input{5new}

\section{Discussion}
\input{6conclusion}
\newpage
\bibliography{gfg} 

\newpage

\appendix
\section{Appendix for Section 3}
\input{71appendix}

\section{Appendix for Section 4}
\input{72appendix}

\section{Appendix for Section 5}
\input{73appendix}

\end{document}

%% file: 0macros.tex
\usepackage{amsfonts}
\usepackage{amsmath}
\usepackage{bm}
\usepackage{amssymb}
\usepackage[pdftex,dvipsnames]{xcolor}
\usepackage{tikzlings}
\usepackage{orcidlink}
\usepackage{tikzducks}
\usepackage{tikzsymbols}
\usepackage{tikz}
\usepackage{mathtools}
\usepackage{thm-restate}
\usepackage{hyperref}
\usepackage[capitalise]{cleveref}
\usepackage{xspace}
\usepackage[new]{old-arrows}
\usepackage{float}
\usepackage{wrapfig}
\usepackage[bold,full]{complexity}
\newcommand{\Func}[1]{{\mathsf{#1}}}
\newcommand{\delay}{\Func{Delay}}
\newcommand{\opt}{\text{opt}}

\newcommand{\eve}{\pmb{\exists}}
\newcommand{\adam}{\pmb{\forall}}
\newcommand{\Gc}{\mathcal{G}} 
\newcommand{\Lc}{\mathcal{L}}
\newcommand{\Ac}{\mathcal{A}}
\newcommand{\Fc}{\mathcal{F}}
\newcommand{\Hc}{\mathcal{H}}

\newcommand{\Bc}{\mathcal{B}}
\newcommand{\Dc}{\mathcal{D}}
\newcommand{\Mc}{\mathcal{M}}

\newcommand{\Sc}{\mathcal{S}}

\newcommand{\Cc}{\mathcal{C}}
\newcommand{\Oc}{\mathcal{O}}
\newcommand{\Pc}{\mathcal{P}}
\newcommand{\Uc}{\mathcal{U}}

\newcommand{\rank}{\text{rank}\xspace}
\newcommand{\stepsim}{\text{step-ahead simulates}\xspace}
\newcommand{\stepsimnoun}{\text{step-ahead simulation}\xspace}
\newcommand{\stepsimgame}{\text{step-ahead simulation game}\xspace}
\newcommand{\racesimgame}{\text{sprint simulation game}\xspace}
\newcommand{\racesim}{\text{sprint simulates}\xspace}
\newcommand{\racesimnoun}{\text{sprint simulation}\xspace}

\newcommand{\raceaheadselfsim}{\text{sprint self-simulation}\xspace}
\newcommand{\lex}{\text{sprint deterministic}\xspace}
\newcommand{\lexnoun}{\text{sprint determinism}\xspace}
\newclass{\TOWER}{TOWER}
\newclass{\ACKERMANN}{ACKERMANN}
\newclass{\EXPTIME}{EXPTIME}
\newclass{\IIEXPTIME}{2-EXPTIME}

\usepackage{graphicx}
\usepackage{array}
\usepackage{tikz}
\usetikzlibrary{automata, arrows, positioning, decorations.markings, decorations.pathreplacing}
\usetikzlibrary {decorations.pathmorphing, decorations.shapes}
\usetikzlibrary{backgrounds,fit,shapes.geometric}
\usetikzlibrary {graphs, quotes}
\usetikzlibrary{calc}
\usetikzlibrary{shapes,shapes.geometric,fit}
\usepackage[new]{old-arrows}
\usetikzlibrary {graphs,shapes.geometric,quotes}
\usetikzlibrary{decorations.pathreplacing}
\usetikzlibrary{automata}
\usetikzlibrary{calc}
\usetikzlibrary{arrows,automata,decorations.pathmorphing}
\usetikzlibrary{shapes,shapes.geometric,fit,calc,automata}
\usepackage{xargs} 
\usepackage{todonotes}
\newcommandx{\ra}[2][1=]{\todo[linecolor=Dandelion,backgroundcolor=Dandelion!25,bordercolor=Dandelion,#1]{\tiny #2-Rohan}}
\newcommandx{\mj}[2][1=]{\todo[linecolor=blue,backgroundcolor=blue!25,bordercolor=blue,#1]{\tiny M: #2}}
\newcommandx{\ap}[2][1=]{\todo[linecolor=red,backgroundcolor=red!25,bordercolor=red,#1]{\tiny #2-Aditya}}

%% file: 1intronew.tex
\noindent 
History-deterministic (HD) automata are non-deterministic automata in which the non-determinism can be resolved ``on the fly'', based only on the prefix of the word read so far~\cite{BL23,Kup22}.
This concept can be formalised using the history-determinism game (HD game), in which two players Adam and Eve make alternating moves choosing letters and transitions, thus constructing a word and a run of the automaton on it, respectively.
Eve wins if the run is accepting or if the word is not in the language, and hence Eve's winning strategy will successfully resolve non-determinism by constructing an accepting run on the fly, for all words in the language.
An automaton is then defined to be history-deterministic if Eve has a winning strategy in the game.

Henzinger and Piterman~\cite{HP06} introduced HD automata because of their potential to speed up key algorithmic tasks in verification and synthesis, such as language containment and strategy synthesis.
In language containment, we ask whether all executions of an implementation~$\Ac$ satisfy a specification~$\Hc$.
If $\Hc$ is non-deterministic then the problem is $\PSPACE$-hard, but if $\Hc$ is HD then it is more tractable, because it amounts to checking that $\Hc$ simulates~$\Ac$,
 which can be done in polynomial time if the parity index of~$\Hc$ is fixed \cite[Theorem 3]{CHP07} and in quasi-polynomial time otherwise~\cite[Theorem 20]{Pra24a}.
Henzinger and Piterman introduced HD automata as good-for-games automata because games whose winning conditions are represented by an HD automaton can be solved efficiently without automaton determinisation~\cite[Theorem 3.1]{HP06}, a well-known compuational bottleneck in synthesis.

\subsection{Related work}

Key questions studied for HD parity automata include recognising them, their succinctness relative to deterministic automata, minimisation, and determinisation.


\subparagraph{Recognising History-Deterministic Automata via Games.}
Kuperberg and Skrzypczak~\cite{KS15} gave a polynomial time algorithm to recognise HD co-B\"uchi automata, and Bagnol and Kuperberg~\cite{BK18} gave a polynomial time algorithm to recognise HD B\"uchi automata. 
These algorithms 
have been conceptually unified by Boker, Kuperberg, Lehtinen, and Skrzypczak~\cite{BKLS20b} to be based on the $2$-token game introduced by Bagnol and Kuperberg~\cite{BK18}, leading to the $2$-token conjecture.


\begin{conjecture}[The $2$-token conjecture \cite{BK18,BKLS20b}]
    A parity automaton is HD if and only if Eve wins the 2-token game on it. 
\end{conjecture}
Proving the $2$-token conjecture would imply that recognising HD parity automata of fixed parity index can be done in polynomial time.
In contrast, the best upper bound currently known for the problem is $\EXPTIME$, dating back to Henzinger and Piterman~\cite{HP06}. In the general case, when the parity index is not fixed, a lower bound of $\NP$-hardness
has been achieved only very recently~\cite{Pra24a}.


Bagnol and Kuperberg~\cite{BK18} introduced the $k$-token game as a tool to characterise the conceptually more complex HD game.
Like in the HD game, in the $k$-token game for $k \geq 1$, two players Adam and Even make alternating moves choosing letters and transitions, but in addition, in every round, after Eve chooses a transition, Adam also chooses $k$ transitions.
As a result, Adam constructs a word and $k$ runs, Eve constructs a run, and Eve wins if her run is accepting or all of Adam's $k$ runs are rejecting.
A key insight in Bagnol and Kuperberg's work~\cite{BK18} is that the $2$-token game is equivalent to the $k$-token game for all $k \geq 2$.

\begin{theorem}[\cite{BK18}]\label{thm:BK18-2tg-equals-ktg}
    Eve wins the $2$-token game on a parity automaton if and only if for all $k \geq 2$, Eve wins the $k$-token game on it.
\end{theorem}
Bagnol and Kuperberg's proof that the $2$-token game characterises history-determinism for B\"uchi automata exploits this insight, by showing that if Adam wins the HD game on a B\"uchi automaton then he can win the $k$-token game for some $k$ that is doubly-exponential in the size of the automaton, and hence also the $2$-token game. 


Boker et al.~\cite{BKLS20b} have used an analogous, but more involved, argument to show that the $2$-token game also characterises history-determinism for co-B\"uchi automata, combining \cref{thm:BK18-2tg-equals-ktg} with the algorithm of Kuperberg and Skrzypczak~\cite{KS15} to recognise HD co-B\"uchi automata efficiently, which was based on the so-called Joker game.
The Joker game is similar to the $1$-token game but, additionally, Adam has the power to (finitely many times) ``play Joker'' by choosing a transition from Eve's token instead of a transition from his token, and Eve wins if her run is accepting, or Adam's run is rejecting, or Adam has played Joker infinitely many times. 


Kuperberg and Skrzypczak's algorithm uses Joker games in their polynomial time algorithm to recognise HD co-B\"uchi automata, but to date, it was not known if Joker games characterise history-determinism on co-B\"uchi or B\"uchi automata, or on parity automata in general.

\subparagraph{Succinctness and minimisation of HD automata.}

Kuperberg and Skrzypczak~\cite{KS15} proved that HD co-B\"uchi automata are exponentially more succinct than deterministic co-B\"uchi automata~\cite{KS15}, which is known to be tight \cite[Theorem 4.1]{Kup22}. 
Abu Radi and Kupferman~\cite{RK22} showed that transition-based HD co-B\"uchi automata can be minimised in polynomial time and that they have canonicity. 
In contrast,  minimisation of state-based HD B\"uchi or HD co-B\"uchi automata is known to be $\NP$-complete~\cite[Theorem 1]{Sch20}.
Minimisation of transition-based B\"uchi automata is easily seen to be in $\NP$, but the exact complexity is open.

\subparagraph{Determinisation of HD B\"uchi automata.}
Kuperberg and Skrzypczak~\cite{KS15} also proved that every HD B\"uchi automaton with $n$ states has an equivalent deterministic B\"uchi automaton with at most $n^2$ states. 
However, it is not known if HD B\"uchi automata are strictly more succinct than deterministic B\"uchi automata. 

The determinisation procedure of Kuperberg and Skrzypczak for an HD B\"uchi automaton $\Hc$ involves carefully analysing the simulation game between $\Hc$ and an equivalent deterministic B\"uchi automaton of exponential size.  
At a high level, the procedure iteratively modifies the simulation game and the automaton~$\Hc$, eventually yielding an equivalent game of quadratic size, from which a deterministic B\"uchi automaton of quadratic size can be extracted, but the procedure itself runs in exponential time. 

Kuperberg and Skrzypczak also gave a non-deterministic polynomial-time procedure for determinisation of HD B\"uchi automata, which guesses a deterministic B\"uchi automaton of quadratic size and then checks for language equivalence~\cite[Theorem 10]{KS15}. 
They left the exact complexity of determinisation for HD B\"uchi automata open, in particular, the question of whether HD B\"uchi automata can be determinised in polynomial time.



\subsection{Our Contributions}
\begin{itemize}
    \item 
        We introduce the $k$-lookahead game, a variant of the $1$-token game, in which Adam's transition on his token is delayed by $k$ steps, thus giving him a lookahead of $k$. 
        We prove that the $1$-token game is equivalent to the $k$-lookahead game. 
        \begin{restatable}{mainthm}{thmlookahead}\label{mainthm:lookahead}\label{theorem:lookahead-games}
            For every parity automaton $\Ac$, Eve wins the $1$-lookahead game on $\Ac$ if and only if she wins the $k$-lookahead game on $\Ac$. 
        \end{restatable}
        The $1$-token game is syntactically equivalent to the $1$-lookahead game. 
        \cref{mainthm:lookahead} thus demonstrates that the $1$-token game is already quite powerful, and it is analogous to \cref{thm:BK18-2tg-equals-ktg} of Bagnol and Kuperberg. 

    \item 
        With \cref{mainthm:lookahead} as a key tool, we show that the $1$-token game characterises history-determinism on semantically-deterministic B\"uchi automata. 
        These are automata in which, for every state, all transitions labelled by the same letter lead to language-equivalent states~\cite{AK23}. 
        \begin{restatable}{mainthm}{thmSDbuchiautomata}\label{mainthm:sdbuchiautomata}\label{thm:main}
            A semantically-deterministic B\"uchi automaton is history-determinsitic if and only if Eve wins the $1$-token game on it. 
        \end{restatable}
        A consequence of \cref{mainthm:sdbuchiautomata} is that Joker games characterise history-determinism on B\"uchi automata (\cref{cor:joker-game}). Since Joker games have smaller arenas than 2-token games, this leads to a more efficient algorithm for recognising HD B\"uchi automata (\cref{lemma:joker-game-complexity}).

    \item  
        We give a parity automaton with priorities $1$, $2$, and~$3$ on which Eve wins the Joker game but that is not HD (\cref{thm:joker-games-counterexample}). 
        This implies that the Joker game does not characterise history-determinism for parity automata and that \cref{mainthm:sdbuchiautomata} does not extend to parity automata.

    \item 
        We give a polynomial time determinisation procedure for HD B\"uchi automata, thus resolving an open question of Kuperberg and Skrzypczak~\cite{KS15}. 
        \begin{restatable}{mainthm}{thmefficientdeterminisation}\label{mainthm:efficientdeterminisation}\label{theorem:efficient-determinisation}
            There is a polynomial-time procedure that converts every HD B\"uchi automaton with $n$ states into an equivalent deterministic B\"uchi automaton with $n^2$ states.
        \end{restatable}
        Our determinisation procedure is inspired by that of Kuperberg and Skrzypczak~\cite{KS15}, but rather than working with the simulation game between the automaton and a deterministic automaton of exponential size, thanks to \cref{mainthm:sdbuchiautomata}, we can work with the $1$-token game instead.  
        This results in an algorithm that is conceptually simpler and that runs in polynomial time.

    \item 
        We also give a technique to reduce game-based characterisations of history-determinism to universal automata (automata that accept all words). 
        Hence to prove the 1-token game characterisation of history-determinism for semantically-deterministic (SD) B\"uchi automata, it suffices to prove it for universal SD B\"uchi automata (\cref{thm:universal-g1}). 
        Likewise, to prove the $2$-token conjecture for parity automata, it suffices to prove it for universal parity automata  (\cref{thm:universal-g2}).
\end{itemize}

%% file: 2prelimsnew.tex
We let $\mathbb{N}=\{0,1,2,\cdots \}$ to be the set of natural numbers. For two natural numbers $i,j$ such that $i<j$, we write $[i,j]$ to denote the set of integers $\{i,i+1,\dots,j\}$, and $[j]$ to denote $[0,j]$. An \textit{alphabet} $\Sigma$ is a finite set of \textit{letters}. We use $\Sigma^{*}$ and $\Sigma^{\omega}$ to denote the set of words of finite and countably infinite length over $\Sigma$, respectively. We also let $\varepsilon$ be the unique word of length~$0$. A language $\Lc \in \Sigma^{\omega}$ is a set of words. For a finite word $u \in \Sigma^{*}$ and a language $\Lc$, we define $u^{-1}\Lc$ to be $\{w\mid uw \in \Lc \}$.

\subsection{Games}

\subparagraph{Game arenas.}
An \emph{arena} is a directed graph $G=(V,E)$ with vertices partitioned into $V_{\adam}$ and $V_{\eve}$ between two players Adam and Eve, respectively. 
Additionally, a vertex $v_0 \in V_{\adam}$ is designated as the initial vertex. 
We say that vertices in $V_{\eve}$ are owned by Eve and those in $V_{\adam}$ are owned by Adam. 

A \emph{play} on this arena is an infinite path starting at $v_0$ and it is formed as follows. 
A play starts with a token at $v_0$ and it proceeds for infinitely many rounds. 
At each round, the player who owns the vertex on which the token is currently placed chooses an outgoing edge, and the token is moved along this edge to the next vertex for another round of play. 
This creates an infinite path in the arena, which we call a play of $G$. 

A \emph{game} $\Gc$ consists of an arena $G=(V,E)$ and a winning condition given by a language $L \subseteq E^{\omega}$. 
We say that Eve \emph{wins a play} $\rho$ in $G$ if $\rho \in L$, and Adam wins otherwise. A \emph{strategy} for Eve in such a game $\Gc$ is a function from the set of plays that end at an Eve's vertex to an outgoing edge from that vertex.
Eve's strategy is said to be winning if any play produced while she plays according to this strategy is winning for her. 
We say that Eve \emph{wins the game} if she has a winning strategy. Winning strategies are defined for Adam analogously, and we say that Adam wins the game if he has a winning strategy. In this paper we deal with $\omega$-regular games, which are known to be determined~\cite{Mar75,GH82}, i.e., each game has a winner. Two games are \emph{equivalent} if they have the same winner.

\subparagraph{Parity games.}
An $[i,j]$-\emph{parity game} $\Gc$ is played over a finite game arena $G = (V,E)$, with the edges of $G$ labelled by a priority function $\chi:E \xrightarrow{} [i,j]$ for some $i,j \in \mathbb{N}$ with $i<j$, and $i=0$ or $i=1$. 
A play $\rho$ in the arena of $\Gc$ is winning for Eve if and only if the highest edge priority that occurs infinitely often is even. 
It is well known that parity games can be solved in polynomial time when the interval $[i,j]$ is fixed, and in quasi-polynomial time otherwise~\cite{CaludeJKLS22,JL17,LPSW22}.

\subsection{Automata}

\subparagraph*{Parity automata.}
An $[i,j]$-\emph{non-deterministic parity automaton} $\Ac = (Q,\Sigma,q_0,\Delta)$ consists of a finite directed graph with edges labelled by letters in $\Sigma$ and \emph{priorities} in $[i,j]$ for some $i,j \in \mathbb{N}$ with $i<j$. 
These edges are called \emph{transitions}, which are elements of the set $\Delta \subseteq Q \times \Sigma \times [i,j] \times Q$, and the vertices of this graph are called \emph{states},  which are elements of the set $Q$. Each automaton has a designated \emph{initial state} $q_0 \in Q$. 
For states $p,q \in Q$ and a letter $a \in \Sigma$, we use $p\xrightarrow{a:c}q$ to denote a transition from $p$ to $q$ on the letter $a$ that has the priority $c$. 

A \emph{run} on an infinite word $w$ in $\Sigma^{\omega}$ is an infinite path in the automaton, starting at the initial state and following transitions that correspond to the letters of $w$ in sequence.  We write that such a run is \emph{accepting} if it \emph{satisfies the parity condition}, i.e., the highest priority occurring infinitely often amongst the transitions of the run is even, and a word $w$ in $\Sigma^{\omega}$ is accepting if the automaton has an accepting run on $w$. The \emph{language} of an automaton $\Ac$, denoted by $\Lc(\Ac)$, is the set of words that it accepts. We write that the automaton $\Ac$ \emph{recognises} a language $\Lc$ if $\Lc(\Ac)=\Lc$. A language $\Lc \subseteq \Sigma^{\omega}$ is \emph{$\omega$-regular} if it is recognised by some parity automaton. A parity automaton $\Ac$ is said to be \emph{deterministic} if for any given state in $\Ac$ and any given letter in $\Sigma$, there is at most one outgoing transition from that state on that letter.  

We write that $[i,j]$, with $i=0$ or $1$, is the parity index of $\Ac$. A \emph{B\"uchi} (resp. co-B\"uchi) automaton is a $[1,2]$ (resp. $[0,1]$) parity automaton. A \emph{safety automaton} is one where all transitions have priority $0$.

We write $(\Ac,q)$ to denote the automaton $\Ac$ with $q$ as it's initial state, and $\Lc(\Ac,q)$ to denote the language it recognises. Two states $p$ and $q$ in $\Ac$ are \emph{equivalent} if $\Lc(\Ac,p) = \Lc(\Ac,q)$.

\subparagraph{History-determinism.}
 The (HD game) is a two player turn-based game between Adam and Eve, who take alternating turns to select a letter and a transition in the automaton (on that letter), respectively. After the game ends, the sequence of Adam's choices of letters is an infinite word, and the sequence of Eve's choices of transitions is a run on that word. Eve wins the game if her run is accepting or Adam's word is rejecting, and we say that an automaton is HD if Eve has a winning strategy in the history-determinism game.

\begin{definition}[History-determinism game]
Given a parity automaton $\Ac = (Q,\Sigma, q_0, \Delta)$, the \emph{history-determinism game} of $\Ac$ is defined between the players Adam and Eve as follows, with positions in $Q$. The game starts at $q_0$ and proceeds for infinitely many rounds. 
For each $i \in \mathbb{N}$, round $i$  starts at a position $q_i \in Q$, and proceeds as follows: \begin{enumerate}\item Adam selects a letter $a_i \in \Sigma$;
\item Eve selects a transition $q_i \xrightarrow{a_i:c_i} q_{i+1} \in \Delta$.    
\end{enumerate}  
The new position is $q_{i+1}$ from where round $(i+1)$ is played. 
Thus, the play of a history-determinism game can be seen as Adam constructing a word letter-by-letter, and Eve constructing a run transition-by-transition on the same word. Eve wins such a play if the following holds: if Adam's word  is in $\Lc(\Ac)$, then Eve's run is accepting.
\end{definition}
If Eve wins the history-determinism game on $\Ac$, then we say that $\Ac$ is \emph{history-deterministic}. 

\subparagraph{Semantic-determinism.}
Let $\Ac$ be a parity automaton.  A transition $\delta$ from $p$ to $q$ on a letter $a$ in $\Ac$ is called \emph{language-preserving} if $\Lc(\Ac,q) = a^{-1} \Lc(\Ac,p)$. We say that a parity automaton  is \emph{semantically-deterministic}, SD for short, if all transitions in it are language-preserving. The following lemma can be shown by a simple inductive argument on the length of words.
\begin{lemma}\label{lemma:SDautomata}
If a parity automaton is SD then all states in the automaton that can be reached from a fixed state $q$ upon reading a finite word $u$ accept the language $u^{-1}\Lc(\Ac,q)$.
\end{lemma}
SD automata were introduced by Kuperberg and Skrzypczak as residual automata~\cite{KS15}. 
We follow Abu Radi, Kupferman, and Leshkowitz~\cite{RKL21} by calling them SD automata instead.

\subsection{Games on Automata}

Simulation and simulation-like games (such as token games~\cite{BK18}) are fundamental amongst the techniques we use in this paper. We define these games below.  

\subparagraph{Simulation and stepahead simulation.}\label{subsec:sim-game}
 \begin{definition}[Simulation game]
      Given two parity automata $\Ac=(P,\Sigma,p_0,\Delta_A)$ and $\Bc=(Q,\Sigma,q_0,\Delta_B)$, the \emph{simulation game} between $\Bc$ and $\Ac$ is a two player game played between Eve and Adam as follows, with positions in $P \times Q$. The game starts at $(p_0,q_0)$, and proceeds for infinitely many rounds. For each $i \geq 0$, round $i$ starts at position $(p_i,q_i)$ and proceeds as follows:
    \begin{enumerate}
            \item Adam selects a letter $a_i \in \Sigma$;
        \item Adam selects a transition $p_i \xrightarrow{a_i:c'} p_{i+1}$ in $\Ac$;
        \item Eve selects a transition $q_i \xrightarrow{a_i:c} q_{i+1}$ in $\Bc$.
    \end{enumerate}

At the end of a play of the simulation game, the letters selected by Adam in sequence form a word, while the sequence of his selected transitions and the sequence of Eve's selected transitions form runs on that word in $\Ac$ and $\Bc$, respectively. We say Eve wins the game if her run in $\Bc$ is accepting or Adam's run in $\Ac$ is rejecting. If Eve has a strategy to win the simulation game, then we say that $\Bc$ \emph{simulates} $\Ac$.
 \end{definition}

The \emph{stepahead simulation} game between $\Bc$ and $\Ac$ is defined similarly to the simulation game, except the orders of move in each round are changed as follows: Adam selects a letter first, then Eve selects a transition on $\Bc$, and then Adam selects a transition on $\Ac$. The winning condition is identical, which is that Eve's run on $\Bc$ is accepting if Adam's run on $\Ac$ is accepting. If Eve wins the stepahead simulation game between $\Bc$ and $\Ac$, then we say that $\Bc$ \emph{\stepsim} $\Ac$.

\subparagraph{Token games.}\label{sec:token-games}
$k$-token games are similar to stepahead simulation games, but are played on a single automaton, and Adam constructs $k$ runs instead of one. The winning objective of Eve requires her to construct an accepting run if one of $k$ Adam's runs is accepting.

\begin{definition}[$k$-token game]
    The \emph{$k$-token game} on a non-deterministic parity automaton $\Ac=(Q,\Sigma,q_0,\Delta)$ is defined between the players Adam and Eve as follows, with positions in $Q \times Q^k$. The game starts at $(q_0,(q_0)^k)$ and proceeds in $\omega$ many rounds. For each $i \in \mathbb{N}$, the  round $i$ starts at a position $(q_i,(p^1_i,p^2_i,\cdots,p^k_i)) \in Q \times Q^k$, and proceeds as follows.\begin{enumerate}
    \item Adam selects a letter $a_i \in \Sigma$.
    \item Eve selects a transition $q_i \xrightarrow{a_i:c} q_{i+1} \in \Delta$.
    \item Adam selects $k$ transitions $p^1_i \xrightarrow{a_i:c'_1} p^1_{i+1}, p^2_i \xrightarrow{a_i:c'_2} p^2_{i+1}, \cdots p^k_i \xrightarrow{a_i:c'_k} p^k_{i+1} \in \Delta$.
\end{enumerate}
The new position is $(q_{i+1},(p^1_{i+1},p^2_{i+1},\cdots,p^k_{i+1}))$, from where round $(i+1)$ begins.

Thus, in  a play of the $k$-token game on $\Ac$, Eve constructs a run and Adam $k$ runs, all on the same word. Eve wins such a play if the following holds: if one of Adam's $k$ runs is accepting, then Eve's run is accepting.
\end{definition}
Observe that the stepahead-simulation game between $\Ac$ and itself is equivalent to the $1$-token game on $\Ac$.

\subparagraph{Joker games.}
\emph{Joker games} are defined similar to $1$-token games, but additionally in each round, Adam can choose to play \emph{Joker} and choose a transition from Eve's position instead of a transition from his position. The winning condition for Eve is the following: If Adam's sequence of transitions satisfies the parity conditions and Adam has played finitely many Jokers, then Eve's run is accepting as well.

\begin{definition}[Joker games]
    The \emph{Joker game} on a non-deterministic parity automaton $\Ac=(Q,\Sigma,q_0,\Delta)$ is defined between the players Adam and Eve as follows, with positions in $Q \times Qk$. The game starts at $(q_0,q_0)$ and proceeds in $\omega$ many rounds. For each $i \in \mathbb{N}$, the  round $i$ starts at a position $(q_i,p_i) \in Q \times Q$, and proceeds as follows.\begin{enumerate}
    \item Adam selects a letter $a_i \in \Sigma$.
    \item Eve selects a transition $q_i \xrightarrow{a_i:c_i} q_{i+1} \in \Delta$
    \item Adam either selects a transitions $p_i \xrightarrow{a_i:c'_i} p_{i+1}$, or plays Joker and selects a transition $q_i \xrightarrow{a_i:c_i} p'_{i+1}$.
\end{enumerate}
The new position is $(q_{i+1},p_{i+1})$, from where round $(i+1)$ begins.

Eve wins such a play if the following holds: if Adam plays finitely many Jokers and his sequence of transitions satisfies the parity condition, then Eve's run is accepting.
\end{definition}

 The following observations are easy to see.
\begin{lemma}\label{lemma:hd-implies-other-games}
    If $\Ac$ is an HD parity automaton, and if $\Bc$ is a parity automaton such that $\Lc(\Bc) \subseteq \Lc(\Ac)$, then we have:
    \begin{enumerate}
        \item Eve wins the Joker game on $\Ac$ and the $k$-token game on $\Ac$, for all $k \geq 1$;
        \item $\Ac$ simulates and \stepsim $\Bc$.
    \end{enumerate}
\end{lemma}
\begin{proof}
    Fix a winning strategy $\sigma$ for Eve in the HD game of $\Ac$. 
    Consider the strategy for Eve in the above games, in which she follows $\sigma$ based on the letters Adam chooses, ignoring the rest of his moves. 
    Then Eve constructs an accepting run whenever the word constructed by Adam is accepting. 
    In particular, if Adam constructs an accepting run in the $k$-token game or the (stepahead) simulation game, then his word must be in $\Lc(\Ac)$, implying Eve's run is accepting as well. 
    
    Similarly, in a play of the Joker game, suppose Adam plays finitely many Jokers and his sequence of transitions satisfies the parity condition. Let $i$ be the last round where Adam played a Joker. Then there is an accepting run on Adam's word, which can be obtained by concatenating Eve's run until round $(i-1)$ with Adam's run from round $i$. This implies that Adam's word is in $\Lc(\Ac)$, once again implying that Eve's run is accepting.
\end{proof}

\subsection{History-Deterministic Automata and Simulation}
Let $\Lc$ be an $\omega$-regular language, and let $\Fc_\Lc$ be the set of automata which recognise languages that are subsets of $\Lc$, i.e. $\Fc_\Lc = \{\Ac \mid \Lc(\Ac) \subseteq \Lc\}$. Consider the relation $\preceq$ on $\Fc_\Lc$ defined as $\Ac \preceq \Bc$ if $\Bc$ simulates $\Ac$. Then, $\preceq$ is a \emph{preorder}:
\begin{enumerate}
    \item $\preceq$ is \emph{transitive}, i.e., $\Ac \preceq \Bc$ and $\Bc \preceq \Cc$ implies $\Ac \preceq \Cc$,
    \item $\preceq$ is \emph{reflexive}, i.e., $\Ac \preceq \Ac$.
\end{enumerate}
We call $\preceq$ the \emph{simulation preorder}.
\cref{lemma:hd-implies-other-games} implies that any HD automaton $\Hc$ which recognises $\Lc$ is a \emph{greatest} element in $\Fc_\Lc$ with respect to the simulation preorder, i.e., $\Ac \preceq \Hc$ for all $\Ac$ with $\Lc(\Ac) \subseteq \Lc(\Hc)$. We show that the converse also holds.

\begin{lemma}\label{lemma:simpreorder}
    Let $\Lc$ be an $\omega$-regular language, and let $\Fc_\Lc=\{\Ac \mid \Lc(\Ac) \subseteq \Lc\}$. Then, an automaton $\Hc$ in $\Fc_\Lc$ is greatest w.r.t. the simulation preorder if and only if $\Hc$~recognises $\Lc$ and is history-deterministic.
\end{lemma}
\begin{proof}
    We only need to prove the forward implication, since the backward implication follows from \cref{lemma:hd-implies-other-games}. Let $\Hc$ in $\Fc_\Lc$ be such that $\Ac \preceq \Hc$ for all $\Ac$ with $\Lc(\Ac) \subseteq \Lc$. Fix $\Dc$ to be a deterministic parity automaton which recognises $\Lc$ (such a $\Dc$ always exists, see \cite[Theorem 1.19 and~3.11]{Automatabook2001}). Then, in particular, $\Dc \preceq \Hc$. Observe that this implies $\Lc(\Hc)\supseteq \Lc(\Dc)=\Lc$, and since $\Lc(\Hc)\subseteq \Lc$, we get that $\Lc(\Hc) = \Lc$.  
    
    We proceed to show that Eve wins the HD game on $\Hc$. Fix $\sigma$ to be a winning strategy for Eve in the simulation game between $\Hc$ and $\Dc$. We use $\sigma$ to construct a winning strategy in the HD game on $\Hc$ as follows. Eve during the letter game on $\Hc$ keeps a corresponding play of the simulation game between $\Hc$ and $\Dc$, where Adam is playing the same letters as the HD game and choosing the unique transitions available to him. Eve then chooses transitions according to $\sigma$ in the HD game and in the simulation game in her memory. This way, whenever Adam's word $w$ in the HD game is in $\Lc(\Hc) = \Lc$, then the unique run in $\Dc$ on $w$ is accepting, and hence, Eve's run in the HD game on $\Hc$ must be accepting as well. 
\end{proof}

We make explicit a corollary of the above lemma.

\begin{corollary}\label{corollary:simulationofhd-implies-hd}
    If a nondeterministic parity automaton $\Ac$ simulates a language-equivalent history-deterministic automaton $\Hc$, then $\Ac$ is history-deterministic.
\end{corollary}
\begin{proof}
    Let $\Lc(\Ac) = \Lc(\Hc) = \Lc$. From \cref{lemma:simpreorder}, we know that $\Hc$ is a greatest element in $\Fc_\Lc$, and $\Hc \preceq \Ac$ implies that $\Ac$ is a greatest element in $\Fc_\Lc$ as well. It follows from \cref{lemma:simpreorder} that $\Ac$ is history-deterministic.
\end{proof}

%% file: 3universalreduction.tex
In the next section, we will show that $1$-token games characterise history-determinism on semantically-deterministic B\"uchi automata (\cref{thm:main}). In order to show this, we start by reducing this result to the restriction where our automata are universal (\cref{thm:universal-g1}), i.e., recognise all words in the language. A very similar reduction also  shows that proving the $2$-token conjecture for universal parity is sufficient to conclude the $2$-token conjecture for parity automata (\cref{cor:2token-game}). 

\begin{restatable}{theorem}{reductionuniversal}\label{thm:universal-g1}
The following statements are equivalent:
    \begin{enumerate}
        \item For any semantically-deterministic B\"uchi automaton $\Ac$, Eve wins the 1-token game on $\Ac$ if and only if $\Ac$ is history-deterministic.
        \item For any semantically-deterministic B\"uchi automaton $\Uc$ with $\Lc(\Uc) = \Sigma^{\omega}$, Eve wins the 1-token game on $\Uc$ if and only if $\Uc$ is history-deterministic.
    \end{enumerate}
\end{restatable}

We shall use the following fact shown by Boker, Henzinger, Lehtinen, and Prakash to prove \cref{thm:universal-g1}~\cite{BHLP23}.

\begin{restatable}[\cite{BHLP23}]{lemma}{lemmabhlp}\label{lemma:bhlp23}
A non-deterministic parity automaton $\Ac$ is history-deterministic if and only if it simulates all deterministic safety automata $\Sc$ with $\Lc(\Sc) \subseteq \Lc(\Ac)$.
\end{restatable}
For self-containment, we include a proof of \cref{lemma:bhlp23} in the appendix. We now prove~\cref{thm:universal-g1}.    

\begin{proof}[Proof sketch for \cref{thm:universal-g1}]
    It is clear that 1 implies 2. For the other direction, suppose 2 holds. Let $\Ac$ be a semantically-deterministic automaton that is not HD. We will show that Adam wins the $1$-token game on $\Ac$. 
    
    From \cref{lemma:bhlp23}, we know that there is a deterministic safety automaton $\Sc$ such that $\Ac$ does not simulate $\Sc$ and $\Lc(\Sc)\subseteq \Lc(\Ac)$. Consider the product safety automaton $\Pc$ of $\Sc$ and $\Ac$ which recognises the language $\Lc(\Pc)=\Lc(\Sc)$. We then complete $\Pc$ by adding an accepting sink state $f$, and transitions to $f$ from all states $q$ on letters $a$ such that $q$ did not have an outgoing transition on $a$ in $\Pc$. We call this automaton $\Uc$. It is clear that $\Lc(\Uc,p)=\Sigma^{\omega}$ for all states $p$ in $\Uc$, and hence $\Uc$ is SD. We show that Adam wins the HD game on $\Uc$, by using his winning strategy in the simulation game between $\Ac$ and $\Sc$ (recall that $\Pc$ was constructed by taking product of $\Sc$ and $\Ac$). The hypothesis implies that Adam wins the $1$-token game on $\Uc$. We then show that we can adapt a winning strategy for Adam on $1$-token game of $\Uc$ to one for the $1$-token game on $\Ac$, by simply `projecting' his strategy to the $\Ac$ component: note that since $\Sc$ is deterministic, in plays of the $1$-token game on $\Uc$, Eve's and Adam's states have the same $\Sc$-component at the start of each round.
\end{proof}

An almost word-by-word identical proof to above also shows that the $2$-token conjecture can be reduced to the case where the automata are universal.

\begin{theorem}\label{thm:universal-g2}
        The following statements are equivalent:
    \begin{enumerate}
        \item For any non-deterministic parity automaton $\Ac$, Eve wins the 2-token game on $\Ac$ if and only if $\Ac$ is history-deterministic.
        \item For any non-deterministic parity automaton $\Uc$ with $\Lc(\Uc) = \Sigma^{\omega}$, Eve wins the 2-token game on $\Uc$ if and only if $\Uc$ is history-deterministic.
    \end{enumerate}
\end{theorem}

%% file: 4g1forbuchi.tex
In this section, we will show the following result.

\thmSDbuchiautomata*

Towards this, we first introduce $k$-lookahead games, which are variants of 1-token games where Adam is given a lookahead of $k$.

\subsection{Lookahead Games}
Let us briefly recall  how a round of the $1$-token game on a parity automaton $\Ac$ works. In each round, Adam selects a letter, then Eve selects a transition on that letter on her token, and then Adam selects a transition on that letter on his token. The winning condition for Eve is that at the end of the play, either Eve's run is accepting or Adam's run is rejecting. This is very close to the simulation game between $\Ac$ and itself, except that the order of the moves in which Eve and Adam select transitions has been reversed. One can, however, see the $1$-token game as a simulation game, where Adam picks the transition for round $i$ in round ($i+1$). Or equivalently, we can construct an automaton $\delay(\Ac)$ such that any non-determinism on $\Ac$ is `delayed' by one step, and then the $1$-token game on $\Ac$ is equivalent to the simulation game between $\Ac$ and $\delay(\Ac)$.     
This insight was used by Prakash and Thejaswini to give an algorithm for deciding history-determinism of one-counter nets, by reducing the $1$-token game to a simulation game~\cite{PT23}. Below we give a construction $\delay$ on parity automata that delays the non-determinism by one-step, inspired by a similar construction for one-counter nets~\cite[Lemma 11]{PT23}.  

\begin{definition}\label{def:delayinf}
For any non-deterministic parity automaton $\Ac = (Q,\Sigma,q_0,\Delta)$, the automaton $\delay(\Ac)$ is constructed so that it runs `one letter behind' $\Ac$, by storing a letter in its state space. More formally, 
$\delay(\Ac) = (Q',\Sigma,s,\Delta')$, 
where $Q' = Q\times \Sigma \cup \{s\}$, and $s$ is the initial state. The set of transitions $\Delta'$ is the union of the following sets of transitions. \begin{enumerate}
    \item $\{(s\xrightarrow{a:0}(q_0,a)) \mid a \in \Sigma\}$. \item $\{((p,a)\xrightarrow{a':c}(q,a')) \mid \ (p\xrightarrow{a:c}q) \in \Delta \}$.
\end{enumerate}

\end{definition}

Observe that $\delay(\Ac)$ accepts the same language as $\Ac$. The following lemma is easy to prove, since the expanded game arenas of the $1$-token game on $\Ac$, and the simulation game between $\Ac$ and $\delay(\Ac)$ are equivalent, with identical winning conditions.
\begin{lemma}\label{lemma:g1todelaysim}
    For every non-deterministic parity automaton $\Ac$, Eve wins the $1$-token game on $\Ac$ if and only if $\Ac$ simulates $\delay(\Ac)$.
\end{lemma}

Furthermore, we can also show that Eve wins the $1$-token game on $\delay(\Ac)$ if Eve wins the $1$-token game on $\Ac$, by simply `delaying' her winning strategy in the $1$-token game on  $\Ac$. 
\begin{lemma}\label{lemma:g1ondelay}
    If Eve wins the $1$-token game on an automaton $\Ac$, then Eve wins the $1$-token game on $\delay(\Ac)$.
\end{lemma}
An iterative application of ~\cref{lemma:g1todelaysim,lemma:g1ondelay} gives us the following corollary.
\begin{corollary}\label{cor:iterativedelay}
If Eve wins the $1$-token game on a parity automaton $\Ac$, then $\Ac$ simulates $\delay^k(\Ac)$ for all $k\in \mathbb{N}$.
\end{corollary}
\begin{proof}
    Note that simulation relation is transitive, i.e., if $\Ac_0$ simulates $\Ac_1$ and $\Ac_1$ simulates $\Ac_2$, then $\Ac_0$ simulates $\Ac_2$. 
    Suppose Eve wins the $1$-token game on $\Ac$. From~\cref{lemma:g1ondelay}, induction gives us that  Eve wins 
    the 1-token game on $\delay^k(\Ac)$ for all $k \in \mathbb{N}$. From~\cref{lemma:g1todelaysim}, we see that $\delay^{k}(\Ac)$ simulates $\delay^{k+1}(\Ac)$ for all $k \in \mathbb{N}$. 
    Combining this with transitivity of simulation, we get that $\Ac$ simulates $\delay^k(\Ac)$ for all $k \in \mathbb{N}$.
\end{proof}

Call the simulation game between $\Ac$ and $\delay^k(\Ac)$ as the \emph{$k$-lookahead game} on $\Ac$. Note that the $1$-token game of $\Ac$ is then equivalent to the $1$-lookahead game of $\Ac$. \cref{cor:iterativedelay} can thus be restated as the following theorem.

\thmlookahead*

\subsection{Games to Characterise History-Determinism}

We now proceed to show \cref{thm:main}. From \cref{thm:universal-g1}, we know that it suffices to only consider SD B\"uchi automata that are universal. 
The following lemma shows that every universal SD B\"uchi automaton is history-deterministic with sufficient lookahead.

\begin{restatable}{lemma}{universaldelay}\label{lemma:universal-delay-hd}
    Let $\Uc$ be a semantically-deterministic B\"uchi automata such that $\Lc(\Uc)=\Sigma^{\omega}$. Then, there is a $K$ such that $\delay^K(\Uc)$ is history-deterministic.
\end{restatable}
\begin{proof}[Proof sketch]  We let $K=2^n$, where $n$ is the number of states of $\Ac$. The crucial observation is that since $\Lc(\Uc,q)=\Sigma^{\omega}$ for any state $q$ in $\Uc$, any finite word $u$ that has length at least $2^n$ must have a run from $q$ that passes through an accepting transition. Eve thus wins the history-determinism game on $\Uc$ by exploiting the lookahead in $\delay^K(\Ac)$ to take at least one accepting transition every $K$ steps. The run of Eve's token then has infinitely many accepting transitions and hence is accepting, as desired.
\end{proof}

We can now prove \cref{thm:main}. 
\begin{proof}[Proof of \cref{thm:main}]
    The forward implication is clear by \cref{lemma:hd-implies-other-games}. For the backward direction, suppose that Eve wins the 1-token game on $\Ac$. Due to \cref{thm:universal-g1}, we may assume that $\Ac$ is universal. 
    Then, from \cref{lemma:universal-delay-hd}, we know that there is a $K$ such that $\delay^K(\Ac)$ is history-deterministic. 
    If Eve wins the 1-token game on $\Ac$, then from \cref{lemma:g1todelaysim}, we know that $\Ac$ simulates $\delay^K(\Ac)$. 
    But since $\delay^K(\Ac)$ is language equivalent to $\Ac$, we get from \cref{corollary:simulationofhd-implies-hd} that $\Ac$ is HD as well.
\end{proof}

Having shown that the $1$-token game on a semantically-deterministic B\"uchi automaton $\Ac$ is equivalent to the HD game on $\Ac$, we are able to show that Joker games characterise history-determinism on B\"uchi automata. We also get an alternate proof of Bagnol and Kuperberg's result of $2$-token games characterising history-determinism of B\"uchi automata as a corollary. 

\begin{restatable}{theorem}{thmcorjokergame}\label{cor:joker-game}\label{cor:2token-game}
    For every B\"uchi automaton $\Ac$, the following statements are equivalent.
        \begin{enumerate}
            \item $\Ac$ is history-deterministic.
            \item Eve wins the Joker game on $\Ac$.
            \item Eve wins the $2$-token game on~$\Ac$.
        \end{enumerate}
\end{restatable}
We prove \cref{cor:2token-game} by reducing it to semantically-deterministic automata, similar to Lemma 16 of Bagnol and Kuperberg (\cref{appendix:b2}).  

Joker games on a B\"uchi automaton have smaller arenas than $2$-token games. As a result, we get a more efficient algorithm to recognise HD B\"uchi automata.
\begin{restatable}{lemma}{jokergamecomplexity}\label{lemma:joker-game-complexity}
     Given a non-deterministic B\"uchi automaton $\Ac=(Q,\Sigma,q_0,\Delta)$, we can decide whether it is history-deterministic in time $\Oc(|\Sigma|^2|Q|^3|\Delta|)$.
\end{restatable}

\begin{remark}
While we have proven \cref{thm:main} by reducing it to universal automata for the sake of simplicity of arguments, one can also give a more direct proof. Such a direct proof involves arguing that if an automaton $\Ac$ is not history-deterministic, then there is a $K$ exponential in the size of $\Ac$ such that $\Ac$ does not simulate $\delay^K(\Ac)$. It then follows from \cref{lemma:g1todelaysim} that Eve loses the $1$-token game on $\Ac$, as desired. 

To argue that $\Ac$ does not simulate $\delay^K(\Ac)$, we reason based on the size of Adam's strategy in the history-determinism game on $\Ac$. Adam, in the simulation game between $\Ac$ and $\delay^K(\Ac)$, can pick letters according to his strategy in the history-determinism game on $\Ac$, thus ensuring Eve produces a rejecting run on her token. At the same time, Adam can exploit the lookahead in $\delay^K(\Ac)$ to construct an accepting run on his token, thus winning the simulation game between $\Ac$ and $\delay^K(\Ac)$. 
\end{remark}

%% file: 4jokercounterexample.tex
\input{figures/fig-joker-game-counterexample}

We end this section by showing that unlike B\"uchi automata, Joker games do not characterise history-determinism on parity automata. 

\begin{example}
 Consider the $[1,3]$-automaton $\Ac$ shown in \cref{fig:counterexample}. Note that $\Ac$ accepts all words in $\{a,b\}^{\omega}$, and is SD. It is easy to see that the automaton $\Ac$ is not history-deterministic, since Adam can win the history-determinism game on $\Ac$ by choosing the letter $a$ when Eve's token is at $p$, and $b$ when her token is at $q$. This forces Eve to never see an even transition, causing her run to be rejecting. 
 
 Eve wins the Joker game on $\Ac$, however. Consider the strategy of Eve where she switches states if her and Adam's tokens are at different states, and otherwise stays at the same state. For Adam to win, Adam must play only finitely many Jokers, and construct an accepting run, which requires him to eventually stay in the same state. But then, Eve's strategy will ensure that Eve's and Adam's run are eventually identical, thus ensuring Adam cannot win the Joker game of $\Ac$. 
\end{example}

\begin{theorem}\label{thm:joker-games-counterexample}
Joker games do not characterise history-determinism on (semantically-deterministic) parity automata.
\end{theorem}



%% file: figures/fig-joker-game-counterexample.tex
\begin{figure}{\centering
 \begin{tikzpicture}[shorten >=1pt,node distance=4cm,on grid,auto,  arrows = {[scale=3]}]

  \node[state,initial, initial text={}]  (s)                {$p$};
  \node[state]          (t) [right=of s] {$q$};

  \path[->] (s) edge   [ bend left]           node        {a,b:3} (t)
                  edge [loop above] node        {b:2} ()
                  edge [loop below] node        {a:1} ()
            (t) edge     [ bend left]         node        {a,b:3} (s)
                  edge [loop above] node        {a:2} ()
                  edge [loop below] node        {b:1} ();
\end{tikzpicture}
\caption{An [1,2,3]-automaton $\Ac$ that is not HD but on which Eve wins the Joker game.}
\label{fig:counterexample}}
\end{figure}
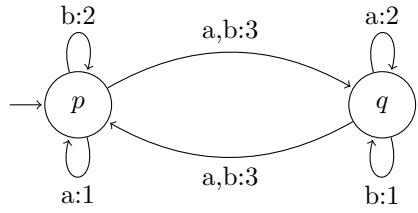

%% file: 5new.tex
In this section, we present a polynomial time determinisation procedure for HD B\"uchi automata with only a quadratic state-space blowup. 
Our procedure combines ideas from the exponential-time procedure given by Kuperberg and Skrzypczak~\cite{KS15} with our 1-token game characterisation of history-determinism on SD B\"uchi automata. 

\thmefficientdeterminisation*

Let $\Hc$ be an HD B\"uchi automaton. Kuperberg and Skrzypczak's procedure relies on carefully analysing the simulation game between $\Hc$ and an equivalent deterministic B\"uchi automaton $\Dc$ whose states are the subsets of $\Hc$. Their procedure iteratively makes modifications to $\Hc$ and $\Dc$ based on the structure of Eve's winning strategies in this game or,
more precisely, the progress measures~\cite{Jur00} or the signatures~\cite{Wal02} of the vertices in this game, which they call ranks. The end result of their procedure is a 
relation
between states of $\Hc$ and states of $\Dc$ that correspond to sets containing a singleton in $\Hc$. A product construction based on this relation allows them to construct an equivalent determinsitic B\"uchi automaton whose number of states is quadratic in the number of states of $\Ac$.

The reason Kuperberg and Skrzypczak work with the simulation game between $\Hc$ and $\Dc$ is that 
it
characterises the history-determinism of $\Hc$. 
A non-deterministic automaton is HD if and only if it simulates an equivalent deterministic one, as shown by Hezinger and Piterman~\cite{HP06}. 
Our result on 1-token games characterising history-determinism (\cref{thm:main}) allows us to instead work with the $1$-token game on $\Hc$. 
This results in a conceptually simpler procedure that works in polynomial time.

We present our algorithm in two steps. 
First, we introduce HD B\"uchi automata that have a \raceaheadselfsim, and we give a polynomial time determinisation procedure for them. 
The procedure involves a quadratic state-space blowup. 
Then, we give a polynomial time iterative procedure to transform $\Hc$ into an equivalent HD B\"uchi automaton of the same size that has a \raceaheadselfsim. 
Overall, this gives us a polynomial time procedure to determinise HD B\"uchi automata with a quadratic state-space blowup.

\subsection{Determinising Automata with Sprint Self-Simulation}\label{subsec:normalised automata}

In this subsection, we describe what automata with \raceaheadselfsim are and give a polynomial time determinisation procedure for such HD B\"uchi automata. The key concept here is the \emph{sprint simulation} relation between two B\"uchi automata, characterised by the \emph{sprint step-ahead simulation game}. This game is similar to the \stepsimgame, but the winning condition for Eve is that she must see an accepting transition before Adam does, i.e., she is in a sprint with Adam to see an accepting transition in the \stepsimgame first. 
Sprint step-ahead simulation (game) would be a more accurate phrase, but for brevity, we shorten it to just sprint simulation (game). 
\begin{restatable}{definition}{cldeflocaldualsimulation}\label{definition:local-dual-simulation}
    For two non-deterministic B\"uchi automata $\Ac = (Q,\Sigma,q_0,\Delta_A)$ and $\Bc$~$=$~$(P,\Sigma,p_0,\Delta_B)$, the \emph{sprint simulation game} between $\Ac$ and $\Bc$
    is played on the set of positions $Q \times P$ and it proceeds in rounds. 
    In each round $i = 0, 1, 2, \dots$, from position $(q_i, p_i)$, the two players Adam and Eve make the following choices:
        \begin{enumerate}
           \item Adam chooses a letter $a_i \in \Sigma$;
           \item Eve chooses a transition $\delta_i=(q_i \xrightarrow{a_i:c_i} q_{i+1}) \in \Delta_A$;
           \item Adam chooses a transition $\delta'_i=(p_i \xrightarrow{a_i:c'_i} p_{i+1}) \in \Delta_B$.
        \end{enumerate}

    The new position is $(q_{i+1},p_{i+1})$. 
    At every round~$i$, if transition $\delta_i$ is accepting then Eve wins the game, and otherwise, if the transition $\delta'_i$ is accepting then Adam wins the game. 
    If neither $\delta_i$ nor $\delta'_i$ are accepting transitions, then the game continues for another round. 
    Eve wins every infinite play.
\end{restatable}
Observe that if Eve wins the \racesimgame then she can do so by a positional strategy, because the objective for Eve is a disjunction of a safety and a reachability objective, which can be seen as a $[0,1]$-parity game. 
If Eve has a winning strategy in the above game, we say that $\Ac$ \racesim $\Bc$.  The sprint simulation relation is transitive, i.e., if $\Ac$ \racesim $\Bc$ and $\Bc$ \racesim $\Cc$, then $\Ac$ \racesim $\Cc$ (\cref{lemma:dual-simulation-transitivity} in appendix). 

\begin{remark}
The sprint simulation relation is similar to the `dependency’ relation introduced by Kuperberg and Skrzypczak~\cite[Definition 30]{KS15}. While the sprint simulation relation is between two B\"uchi automata, dependency relation is derived from the sprint simulation game between a B\"uchi automaton and an equivalent exponential-sized deterministic B\"uchi~automaton.
\end{remark}

We say that an HD B\"uchi automaton $\Hc$ has a \emph{\raceaheadselfsim} if it is semanti\-cally-deterministic and for every state $p$ in $\Hc$, there is a language-equivalent state $q$, such that $(\Hc,p)$ \racesim $(\Hc,q)$. When $\Hc$ is clear from the context, we will just say that $p$ \racesim $q$. 
For the rest of this subsection, fix $\Hc=(Q,\Sigma,q_0,\Delta)$ to be an HD B\"uchi automaton that has a \raceaheadselfsim. 
The following lemma follows from transitivity of \racesimnoun.

\begin{lemma}\label{lemma:dual-simulation-1}
    For every state $p$ in $\Hc$, there is a language-equivalent state $q$ such that $p$ \racesim $q$ and $q$ \racesim itself.
\end{lemma}
\begin{proof}
    Fix a state $p$ in $\Hc$. Then, there is a language-equivalent state $q_0$ in $\Hc$ such that $p$ \racesim $q_0$. 
    If $q_0$ does not \racesim itself, there is another language-equivalent $q_1$ such that $q_0$ \racesim $q_1$. 
    Repeating this argument, we get a sequence of states $q_0,q_1,q_2,\dots$, such that $q_i$ \racesim $q_{i+1}$. Since, there are finitely many states in $\Hc$, we know there are two natural numbers $i<j$ such that $q_i = q_j$. 
    But then, by transitivity of \racesimnoun (\cref{lemma:dual-simulation-transitivity}), we get that $p$ \racesim $q_i$ and $q_i$ \racesim itself, as desired.   
\end{proof}

Let us call a state $q$ \lex if the automaton $\Hc$ can be determinised by deleting transitions to get a deterministic subautomaton $\Fc_q$ so that the following holds: for all finite words $w$ that have a run in $\Hc$ starting at $q$ going through an accepting transition, the unique run on $w$ in $\Fc_q$ from $q$ also sees an accepting transition. Thus, for states $q$ that are \lex, there is a uniform strategy that achieves the objective of seeing an accepting transition as soon as possible on all words. We say that the automaton $\Fc_q$ as above is a witness for \lexnoun of $q$.

\begin{restatable}{lemma}{cllemmadualsimulationtwo}\label{lemma:dual-simulation-2}\label{lemma:plsfix}
    A state $q$ in $\Hc$ is \lex if and only if $q$ \racesim itself. 
    Moreover, there is a deterministic subautomaton $\Fc$ that can be computed in polynomial time and is a witness 
    for all \lex states.
\end{restatable}

\cref{lemma:dual-simulation-1,lemma:dual-simulation-2} above tell us that every state in $\Hc$ is either \lex, or it \racesim a state that is \lex. Fix a subautomaton $\Fc$ from \cref{lemma:dual-simulation-2}, and a positional Eve strategy $\tau$ in the step-ahead simulation game from $(p,q)$ for all pairs of states $(p,q)$ such that $p$ \racesim $q$. 

The deterministic automaton $\Dc$ is then constructed to consist of pairs of states $(p,q)$ such that $p$ and $q$ are language equivalent, $p$ \racesim $q$, and $q$ is sprint deterministic. Note that the pair of such states can be found in polynomial time, since checking for language containment on HD B\"uchi automata~\cite[Corollary 17]{Pra24a} and deciding the winner of sprint simulation game can be done in polynomial time.  Furthermore, for each state $p$ in $\Hc$, we know from \cref{lemma:dual-simulation-1,lemma:dual-simulation-2} that there is a state $q$ in $\Hc$ such that $(p,q)$ is a state in $\Dc$. We let the initial state $d_0$ be $(q_0,r_0)$ for some $r_0$ such that $(q_0,r_0) \in \Dc$. 

At a state $(q,p)$, the transitions in $\Dc$ from the second component $p$ are chosen according to transitions from $\Fc$, while transitions from $q$ are chosen via the positional Eve strategy~$\tau$. 
When an accepting transition $q \xrightarrow{a:2} q'$ is taken on the first component, we update the second component deterministically to be $p'$ such that $(q',p')$ is a state in $\Dc$. 
Or, equivalently, $q'$ and $p'$ are language equivalent, $q'$ \racesim $p'$, and $p'$ is \lex. 
The priorities of transitions in $\Dc$ are the priorities of transitions of the first component. 
A more formal construction of $\Dc$ can be found in the appendix.

We show the correctness of our construction using the definition of sprint simulation game and the fact that $\Hc$ is semantically-deterministic.
\begin{lemma}
    The automaton $\Dc$ accepts the same language as $\Hc$.
\end{lemma}
\begin{proof}
 $\Lc(\Dc) \subseteq \Lc(\Hc)$: 
 If $\rho$ is an accepting run of a word $w$ in $\Dc$, then the projection of $\rho$ on the first component is an accepting run in $\Hc$ as well.

$\Lc(\Hc) \subseteq \Lc(\Dc)$: 
We show that for each state $(p,q)$ in $\Dc$, if $w \in \Lc(\Hc,p)$ then the run from $(p,q)$ on $w$ in $\Dc$ sees an accepting transition eventually. We can then conclude by induction and the semantic determinism of $\Hc$ that a run on $\Dc$ on any word in $\Lc(\Hc)$ contains infinitely many accepting transitions, and hence it is accepting. 
To see this, let $w \in \Lc(\Hc,p)$ and let $\rho_D$ be the run of $\Dc$ on $w$ from $(p,q)$. 
By construction of $\Dc$, we know that $q$ is language-equivalent to $p$.
Since $q$ is \lex, the second component of $\rho_D$ in $\Dc$ must contain an accepting transition on $w$ eventually. But since $p$ \racesim $q$, the run on the first component of $\rho_D$ contains an accepting transition as well, as desired. 
\end{proof}



\subsection{Towards Automata with Sprint Self-Simulation}\label{subsec:normalising-automata}

We now present a polynomial time algorithm to convert an HD B\"uchi automaton into an equivalent HD B\"uchi automaton that has a \raceaheadselfsim. 
Throughout this subsection, let $\Hc=(Q,\Sigma,q_0,\Delta)$ be an HD B\"uchi automaton. 

 We say that $\Hc$ is \emph{good} if $\Hc$ is semantically-deterministic and Eve wins the Joker game from all states in $\Hc$. 
 Every HD B\"uchi automaton $\Hc$ can be converted to an equivalent good HD B\"uchi automaton in polynomial time: we fix a winning strategy $\tau$ for Eve on the Joker game on $\Hc$, and consider the subautomaton $\Hc_N$ consisting of transitions that Eve takes according to $\tau$ (\cref{lemma:good-automata}). We thus assume without loss of generality that $\Hc$ is good.
 
To get an HD B\"uchi automaton equivalent to $\Hc$ and that has a \raceaheadselfsim, we iteratively make modifications to $\Hc$ based on the ranks of the $1$-token game on $\Hc$. We first give a description of the $1$-token game on a B\"uchi automaton as a $[0,2]$-parity game, and briefly recall the properties of ranks that we need on such games.
\begin{definition}
For a semantically-deterministic B\"uchi automaton $\Bc=(Q,\Sigma,q_0,\Delta)$, define the $[0,2]$-parity game $G_1 (\Bc)=(V,E)$ as follows:
\begin{itemize}
    \item The set of vertices $V$ consists of the set $V = V_1 \cup V_2 \cup V_3$, where:
    \begin{enumerate}
        \item $V_1 = \{(p,q) \mid \text{ $p,q$ are states reachable from $q_0$ upon reading the same word $w$}\}$
        \item $V_2 = \{(p,a,q) \mid (p,q) \in V_1 \}$
        \item $V_3 = \{(p',q,a) \mid (p,a,q) \in V_2 \text{ and } p\xrightarrow{a}p' \in \Delta\}$
    \end{enumerate}
    Eve's vertices are $V_{\eve} = V_2$, while Adam's vertices are $V_{\adam} = V_1 \cup V_3$
    \item The set of edges $E$ is the union of following sets:
        \begin{enumerate}
            \item $E_1 = \{(p,q) \xrightarrow{} (p,a,q) \mid a \in \Sigma \}$ (Adam chooses a letter) 
            \item $E_2 = \{(p,a,q) \xrightarrow{} (p',q,a) \mid p\xrightarrow{a:c}p' \in \Delta\} $ (Eve chooses a transition on her token)
            \item $E_3 = \{(p',q,a) \xrightarrow{} (p',q') \mid q \xrightarrow{a:c}q' \in \Delta \}$ (Adam chooses a transition on his token)
        \end{enumerate}
    \item The priority function $\Omega$ is defined as follows.
        All  elements in $E_1$ are assigned priority 0, while edges $(p,a,q) \xrightarrow{} (p',q,a)$ in $E_2$ are assigned priority 2 if the transition $\delta=p \xrightarrow{a:c} p'$ in $\Bc$ is accepting (or equivalently, $c=2$), and 0 otherwise. The edge $(p',q,a) \xrightarrow{} (p',q')$ in $E_3$ is assigned priority 1 if the transition $q\xrightarrow{a:c} q'$ is accepting, and 0 otherwise.
\end{itemize}
\end{definition}
Observe that since $\Hc$ is SD, we have $\Lc(\Hc,p)=\Lc(\Hc,q)$ if $(p,q)$ is a vertex in $G_1(\Hc)$.

\subparagraph{Ranks.} We now define the ranks of a $[0,2]$-parity game $\Gc$. For each vertex $v$ in $\Gc$, $\rank(v)$ is the largest number of 1's that Adam can guarantee Eve will see before seeing a 2 in the play (or only 0's) starting from $v$. 

Observe that Eve wins such a parity game from every position 
if and only if the ranks of all the vertices are bounded. 
If this is the case, then there is a positional winning strategy $\tau$, using which Eve can guarantee that she sees at most $\rank(v)$ many 1's before seeing a 2 (or seeing 0's forever) \cite[Lemma 8]{Wal02} in every play. 
We will call such a strategy \emph{optimal}. 
For $[0,2]$-parity games with $n$ vertices and $m$ edges, an optimal strategy can be computed in time $\Oc(mn)$ \cite[Theorem 11]{Jur00}. 

The following property of ranks follows from their definition.
\begin{lemma}\label{lemma:ranksssss}
    Let $\Gc$ be a $[0,2]$-parity game, and let $v \xrightarrow{e} u$ be an edge in $\Gc$, such that either $v$ 
    belongs to Adam, or the edge $e$ is prescribed by Eve's optimal strategy~$\tau$. 
    Then the edge $e$ has priority 2 or $\rank (v) \geq \rank (u)$. 
    Furthermore, this inequality is strict if $e$ has priority 1.
\end{lemma}

Consider the $1$-token game $G_1(\Hc)$, and the ranks of its vertices. 
Note that for a vertex $(q,p)$ in $G_1(\Hc)$, we have $\rank(q,p)=0$ if and only if $q$ \racesim $p$. 
Define, for each state $q \in \Hc$, its optimal rank $\opt(q)$ to be the minimum rank of a vertex of the form $(q,p)$ in $G_1(\Hc)$. 
Note that if $\opt(q)=0$ for all states $q \in \Hc$, then $\Hc$ has a \raceaheadselfsim. Thus, our iterative procedure focuses on reducing the optimal ranks for all states until they are all $0$. 
We describe this procedure below.

\subparagraph{Iterating towards a sprint self-simulation.}
Set $\Hc_0 = \Hc$. For each $i \geq 0$, perform the following three steps on $\Hc_i$ until $\Hc_{i+1}=\Hc_i$.

\begin{enumerate}
    \item 
        For all vertices $(p,q) \in G_1(\Hc_i)$, compute the optimal ranks $\opt_i(p)$ 
        in $G_1(\Hc_i)$.
    \item 
        Obtain $\Hc'_i$ from~$\Hc_i$ by removing all transitions $q \xrightarrow{a:1}q'$ with $\opt_i(q)<\opt_i(q')$. \label{item:important}
    \item 
        Obtain $\Hc_{i+1}$ from~$\Hc'_i$, by making all transitions $q \xrightarrow{a:1} q'$ with $\opt_i(q)>\opt_i(q')$ accepting.
\end{enumerate}
In the Appendix, we show that for each~$i$, both $\Hc'_i$ and $\Hc_{i+1}$ are good HD B\"uchi automata that are equivalent to $\Hc_i$ (\cref{lemma:iteration-correctness-1,lemma:iteration-correctness-2}). 
By a simple induction, we get that each $\Hc_i$ for $i \geq 0$ is a good HD B\"uchi automaton equivalent to $\Hc$. 

Note that in steps 2 and 3, we are either removing rejecting transitions or making rejecting transitions accepting, and hence this procedure terminates after at most $|\Delta|$ iterations. 
Let $\Hc^{*}$ be the automaton obtained after the procedure terminates. 
Since $\Hc^{*}$ is good, we know that it is SD. The next lemma thus shows that $\Hc^{*}$ has a \raceaheadselfsim.  

\begin{lemma}\label{lemma:h*normalised}
    For all states $p$ in $\Hc^*$, there is a language-equivalent state $q$ in $\Hc^*$ such that $p$ \racesim $q$.
\end{lemma}
\begin{proof}
Assume, to the contrary, that there exists a state $p$ such that $\opt^{*}(p) = \rank^{*}(p,q)>0$. Fix an optimal winning strategy $\tau$ for Eve in $G_1(\Hc^{*})$. Consider a finite play $\rho$ of $G_1(\Hc^{*})$ from $(p,q)$ where Eve is playing according to $\tau$, Adam chooses an accepting transition in his token at some point, while Eve is never able to. Then by monotonicity of ranks (\cref{lemma:ranksssss}), we know that $\rank^{*}$ strictly decreases across $\rho$ at some point. Then, there must be a rejecting transition across which the quantity $\opt^{*}$ decreases as well. But since such transitions would have been made accepting in step 3 of the iteration, we get a contradiction.
\end{proof}
 
From the polynomial time determinisation construction for HD B\"uchi automata that have a \raceaheadselfsim presented in \cref{subsec:normalised automata}, we get a polynomial time determinisation procedure for HD B\"uchi automata. This concludes the proof of \cref{theorem:efficient-determinisation}.

%% file: 6conclusion.tex
Our paper has shown two key results on HD B\"uchi automata: a $1$-token game based characterisation of history-determinism for semantically-deterministic B\"uchi automata, and a polynomial time determinisation procedure. In the process of obtaining these results, we developed several novel techniques that we believe to be equally exciting and insightful. We finish by remarking some implications of our results and techniques, and natural future directions that our work points to. 

Our first technique, presented in \cref{sec:universal-reduction}, reduces game based characterisations of history-determinism on parity automata to parity universal automata. But the history-determinism game on such an automaton is just a parity game, since Adam's word is always accepting. The $2$-token conjecture thus reduces, by \cref{thm:universal-g2}, to showing that this parity game is equivalent to the $2$-token game. This seems easy enough to show at first glance, but it proves to be (unsurprisingly) difficult. This result also shows that the difficulty in proving or disproving the $2$-token conjecture arises not from the language an automaton recognises, but rather from the structure of the automaton.

We also introduced lookahead games, and showed that $k$-lookahead games are equivalent to $1$-token games for all $k \geq 1$ (\cref{mainthm:lookahead}). This shows that the $1$-token games are quite powerful, in the same sense that $2$-token games are powerful due to them being equivalent to $k$-token games for all $k \geq 2$. While our $1$-token game characterisation of history-determinism on SD B\"uchi automata does not extend to parity automata (\cref{thm:joker-games-counterexample}), one can combine the two different approaches to give the $2$-token game more power, both in the form of lookahead and more tokens. It would be interesting to consider such games to try extending the $2$-token conjecture beyond B\"uchi and co-B\"uchi automata.

Our algorithm to determinise HD B\"uchi automata involves a quadratic blowup. However, we do not know whether this is tight. In fact, it is still open if HD B\"uchi automata are strictly more succinct than determinstic B\"uchi automata. Nevertheless, we are hopeful that our algorithm can offer some insights on how to make progress on this problem.

Let us end with a problem highlighted by Boker and Lehtinen in their recent survey \cite[Section 6.3.2]{BL23}. In all the existing game-based characterisations which are used to recognise HD automata efficiently, including ours (\cref{mainthm:sdbuchiautomata}), it is not clear how we can naturally convert a winning strategy for Eve from the $2$-token game or the Joker game to a winning strategy in the HD game. Our algorithm to determinise HD B\"uchi automata, however, can be seen as one: starting with a winning strategy for Eve in the Joker game, we construct a strategy in the HD game that requires linear memory, thus obtaining a deterministic automaton of quadratic size. But the proof of correctness of our algorithm relies on \cref{mainthm:sdbuchiautomata}. Towards a more pure strategy-transfer argument, where ideally an algorithm for strategy transfer also proves a game-based characterisation of history-determinism, our algorithm comes tantalisingly close. Indeed, proving that the automaton constructed in Step \ref{item:important} preserves the relevant invariants (\cref{lemma:iteration-correctness-1}) is the only place where we use the fact that we started with an HD automaton. We believe that trying to get rid of this assumption, in order to give an alternative strategy-transfer proof for the $1$-token game characterisation of history-determinism on SD B\"uchi automata, could lead to crucial insights towards better understanding HD (B\"uchi) automata and token games.  

%% file: 71appendix.tex
\lemmabhlp*
\begin{proof}
Let $\Ac=(Q,\Sigma,q_0,\Delta)$. If $\Ac$ is HD, then by \cref{lemma:hd-implies-other-games}, it simulates all such $\Sc$ with $\Lc(\Sc) \subseteq \Lc(\Ac)$. For the other direction, assume that $\Ac$ is not history-deterministic. Then, Adam has a finite memory winning strategy in the history-determinism game on $\Ac$ which can be represented by a \emph{deterministic transducer} $\Mc=(M,\Delta,m_0,\Delta_M,f:M\xrightarrow{}\Sigma)$ that outputs letters in $\Sigma$ based on transitions that Eve picks in the history-determinism game. More formally, the transducer $\Mc$ corresponds to a finite memory Adam winning strategy as follows. In round $i$ of the HD game on $\Ac$ for each $i \geq 0$, where the Eve's state is at $q_i$ and Adam's memory is at $m_i$, Adam starts by picking the letter $a_i=f(m_i)$. Eve then picks a transition $\delta_i$ in $\Ac_i$ on the letter $a_i$. Adam then updates his memory by the unique transition $m_i \xrightarrow{\delta_i}m_{i+1}$ in $\Mc$. 

Note that $\Mc$ recognises the set of winning plays of Adam, i.e., any word $w_{\Delta} \in \Delta^{\omega}$ that has a run in $\Mc$ represents a rejecting run in $\Ac$ over an accepting word in $\Ac$. 

Consider the safety automaton $\Mc_{\Sigma}=(M,\Sigma,m_0,\Delta'_M)$, obtained by `projecting' $\Mc$ onto $\Sigma$, i.e., $m \xrightarrow{a:0} m'$ is a transition in $\Mc_{\Sigma}$ iff there is a transition $m \xrightarrow{\delta} m'$ such that $f(m)=a$ and $\delta \in \Delta$ is a transition on the letter $a$. Note that $\Lc(\Mc_{\Sigma}) \subseteq \Lc(\Ac)$, since Adam only constructs accepting words in the HD game on $\Ac$ when playing according to $\Mc$.

Now, let $\Sc$ be the safety automaton obtained by determinising $\Mc_{\Sigma}$ via the subset construction. Since $\Lc(\Sc) = \Lc(\Mc_{\Sigma})$, we get that $\Lc(\Sc) \subseteq \Lc(\Ac)$. We claim that $\Ac$ does not simulate $\Sc$. Indeed, consider the following winning strategy of Adam in the simulation game between $\Ac$ and $\Sc$, where Adam uses $\Mc$ as a memory to choose letters. 

At the start of round 0, Adam's token is at the state $\{m_0\}$, his memory is $m_0$, while Eve's token is at $q_0$. At round $i$ for each $i \geq 0$, suppose that Adam's token in at $S_i$, his memory is $m_i$ such that $m_i \in S_i$, and Eve's token is at $q_i$. Then, Adam chooses the letter $a_i=f(m_i)$, Eve picks the transition $\delta_i$ on $a_i$, and Adam's token takes the deterministic transition on $a_i$ to $S_{i+1}\subseteq M$ in $\Sc$. Let $m_{i+1}$ be such that $m_i \xrightarrow{\delta_i} m_{i+1}$ is the the unique transition on $\delta_i$ at $m_i$ in $\Mc$. We update Adam's memory to be $m_{i+1}$, which is an element of $S_{i+1}$.

Note that any Eve's run on her token in the above play would be rejecting, since such plays correspond to a word in $\Mc$. But Adam's run on his token in the safety automaton $\Sc$ is accepting, as desired.
\end{proof}

\reductionuniversal*

\begin{proof}
    It is clear that 1 implies 2. For the other direction, suppose that for all SD B\"uchi automata $\Uc$ with $\Lc(\Uc) = \Sigma^{\omega}$, $\Uc$ is HD if and only if Eve wins the 1-token game on $\Uc$.
    
    Let $\Ac = (Q,\Sigma,\Delta,q_0)$ be a SD B\"uchi automaton. The backward implication is clear, that is, if $\Ac$ is HD, then Eve wins the 1-token game on $\Ac$ (\cref{lemma:hd-implies-other-games}). For the forward implication, suppose $\Ac$ is not HD. Then, by \cref{lemma:bhlp23}, we know there is a deterministic safety automaton $\Sc=(S,\Sigma,s_0,\Delta_S)$ such that $\Lc(\Sc) \subseteq \Lc(\Ac)$, and $\Ac$ does not simulate $\Sc$. For convenience, we will assume that $\Sc$ does not have any states that are not safe, and instead
    might be incomplete, i.e., there can be states from which there is no $a$-transition from some $a \in \Sigma$. Consider the product automaton of $\Ac$ and $\Sc$ given by $\Pc=(P,\Sigma,p_0,\Delta_P)$, where \begin{itemize}
        \item  The set of states $P = Q \times S$, the initial state $p_0 = (q_0,s_0)$
        \item $(q,s) \xrightarrow{a:c} (q',s')$ is a transition in $\Delta_P$ if and only if $q \xrightarrow{a:c}q'$ is a transition in $\Ac$, and $s \xrightarrow{a:0} s'$ is a transition in $\Sc$.
    \end{itemize}
    We note that $\Lc(\Pc) = \Lc(\Sc) \cap \Lc(\Ac) = \Lc(\Sc)$. Consider the automaton $\Uc$ obtained by completing $\Pc$ with an accepting sink state $f$. That is, for all states $p$ in $\Pc$, $a \in \Sigma$ such that there is no outgoing transition from $p$ on $a$, we add the transition $p\xrightarrow{a:k} f$, and the self-loops $f\xrightarrow{a:k}f$ for some even number $k$. In the case of B\"uchi automata, $k=2$.
    \begin{enumerate}
        \item $\Lc(\Uc) = \Sigma^{\omega}$. For any word $w \in \Lc(\Sc)$, it is accepted by $\Pc$ and hence by $\Uc$. If $w \notin \Lc(\Sc)$, then any run of $\Uc$ on $w$ ends up at $f$, and hence $w \in \Lc(\Uc)$. 
        \item $\Uc$ is SD. Indeed, for each state $q$, we can argue similarly that $\Lc(\Uc,q)=\Sigma^{\omega}$.
        \item $\Uc$ is not HD. The player Adam can use a winning strategy for the simulation game between $\Ac$ and $\Sc$, in the HD game on $\Uc$. Since Adam will construct a word in $\Sc$ using this strategy, Eve can never reach $f$, neither can her run in the history-determinism game on $\Uc$ be accepting.
        \item If Adam wins the 1-token game on $\Uc$, then Adam wins the 1-token game on $\Ac$. Let $\sigma$ be a winning strategy for Adam in the 1-token game on $\Uc$. Note that at any point in the 1-token game on $\Uc$, if Eve's token is at a state $(q,s)$, then Adam's token must also be at a state of the form $(q',s)$ since $\Sc$ is deterministic, for some state $q'$ in $\Ac$. Thus, $\sigma$ will never choose a letter $a$ in the $1$-token game on $\Uc$ whenever Eve's token is at a state $(q,s)$  where there is no $a$-transition on $s$ in $\Sc$, since then Eve can move her token to $f$ and win. Similarly, Adam will never be able to move his token to $f$ if he plays according to $\sigma$ as well. With this, we see that any play of the $1$-token game in $\Uc$ where Adam plays according to his $\sigma$ strategy corresponds to a play of $1$-token game on $\Ac$. It follows that Adam wins the $1$-token game on $\Ac$ by playing according to the strategy $\sigma$ and keeping in his memory states of $\Sc$.
        
    \end{enumerate}
     Since Adam wins the 1-token game on $\Uc$ if $\Uc$ is not history-deterministic by the hypothesis, 4 gives us that Adam wins the $1$-token game on $\Ac$, as desired.
\end{proof} 

%% file: 72appendix.tex
\subsection{Lookahead Games}

\universaldelay*
\begin{proof}
    We let $K=2^n$, where $n$ is the number of states in $\Uc$. We start by proving the following claim.
    \begin{claim}\label{claim:universal-delay-lemma}
        For any state $q$ in $\Uc$ that is reachable from the initial state, and any finite word $u$ of length at least $K$, there is a run from $q$ on $u$ in $\Uc$ that visits an accepting transition.
    \end{claim}
    \begin{claimproof}
        Note that since $\Uc$ is semantically-deterministic, we have that $\Lc(\Uc,q)=\Sigma^{\omega}$. 
        Consider the 
        sequence of sets of states that can be visited on reading the prefixes of $u$:
        $$\{q\} \xrightarrow{a_1} S_1 \xrightarrow{a_2} S_1 \cdots \xrightarrow{a_{K}}S_K,$$ where $u=a_1a_2\cdots a_K$ and $S_{i+1}$ is the set of all states that have a transition from $S_i$ on the letter $a_{i}$. Then, observe that two subsets $S_i$ and $S_j$ for $i<j$ must be the same. 
        
        Let $u'$ and $v$ be the finite words  $u'=a_1 a_2 \cdots a_i$ and $v=a_{i+1} a_{i+2} \cdots a_j$. 
        Thus, the sequence of sets of states 
        visited on prefixes of $w$ cycles from $S_i$ upon reading a $v$ to $S_i$. 
        Since $w$ is in $\Lc(\Ac,q)$, there must be an accepting transition seen from a state in $S_i$ on a run on the finite word $v$, as desired. 
    \end{claimproof}
    We proceed to show that $\delay^K(\Ac)$ is history-deterministic, by describing a winning strategy for Eve in the 
    history-determinism game on $\delay^K(\Ac)$. 
    Starting from the initial state in $\delay^K(\Ac)$, the transitions on the first $K$ letters of any word is deterministic, after which Eve's state in $\delay^K(\Ac)$ is of the form $(q_0,u_0)$, where $u_0$ is a word of length $K$. By \cref{claim:universal-delay-lemma}, we know that there is a finite run $\rho_0$ from $q_0$ on $u_0$ that contains at least one accepting transition. Hence, Eve chooses the transitions in $\delay^K(\Ac)$ for the next $K$ rounds in the history-determinism game according to $\rho_0$. After these $K$ rounds, Eve's token in the history-determinism game is at a state $(q_1,u_1)$, where again $u_1$ is a finite word of length $K$. But now Eve can repeat this strategy in the next $K$ rounds, following a run $\rho_1$ from $q_1$ on $u_1$ in $\Ac$ that contains an accepting transition. Thus, Eve repeating this strategy after every $K$ rounds causes Eve to construct an accepting run in the history-determinism game on $\delay^K(\Ac)$. It follows that $\delay^K(\Ac)$ is history-deterministic. 
\end{proof}

\subsection{Joker Games characterises History-Determinism on B\"uchi Automata}\label{appendix:b2}

\thmcorjokergame*

In order to show \cref{cor:2token-game}, we start by showing the following lemma, which allows us to reduce it to the case when $\Ac$ is semantically-deterministic. This is an adaptation of Lemma 16 in the work of Bagnol and Kuperberg~\cite{BK18}, for showing that Eve winning 2-token games on B\"uchi automata characterises history-determinism.  

\begin{lemma}\label{lemma:BK18residual}
Suppose Eve wins the Joker game on a non-deterministic parity automaton $\Ac$, and let $\Bc$ be the sub-automaton obtained by removing any transitions that are not language-preserving. Then we have:
\begin{enumerate}
    \item $\Lc(\Ac) = \Lc(\Bc)$
    \item Eve wins the Joker game on $\Bc$
    \item If $\Bc$ is history-deterministic, then so is $\Ac$
\end{enumerate}
\end{lemma}
\begin{proof}
If Eve wins the Joker game on $\Ac$, then note that her winning strategy can only take language-preserving transitions: if she picks a transition $p\xrightarrow{a}q$ in the history-determinism game on $\Ac$ such that $\Lc(\Ac,q) \subsetneq a^{-1}\Lc(\Ac,p)$, then Adam can play Joker and construct a run on the word $aw$ where $w \in a^{-1}(\Lc(\Ac,p) \setminus \Lc(\Ac,q))$, and win the Joker game. This implies that Eve has to take transitions only in $\Bc$ to win the Joker game, which shows 2. The proof of 1 follows as well, since for any word $w \in \Lc(\Ac)$ that Adam plays, Adam can construct an accepting run in his token, forcing Eve to construct an accepting run on $w$ while taking only language-preserving transitions. 

Finally for 3, if $\Bc$ is history-deterministic, Eve can use her strategy in the history-determinism game on $\Bc$ to win the history-determinism game on $\Ac$, since $\Bc$ is a subautomaton of $\Ac$ with $\Lc(\Bc) = \Lc(\Ac)$. 
\end{proof}

\begin{corollary}\label{cor:joker-hd-buchi}
    Given a B\"uchi automaton $\Ac$, Eve wins the Joker game on $\Ac$ if and only if $\Ac$ is history-deterministic.
\end{corollary}
\begin{proof}
    $\Rightarrow:$ If Eve wins the Joker game on $\Ac$, then she wins the Joker game on $\Bc$, where $\Bc$ is the subautomaton obtained by removing transitions in $\Ac$ that are not language-preserving. Note that from~\cref{thm:main}, a semantically-deterministic B\"uchi automaton $\Bc$ is history-deterministic if and only if Eve wins the 1-token game on $\Bc$, which is true since Eve wins the Joker game on $\Bc$. It follows from~\cref{lemma:BK18residual} that $\Ac$ is history-deterministic as well. 
    
    $\Leftarrow:$ This is clear from \cref{lemma:hd-implies-other-games}.
\end{proof}

Replacing Joker game with 2-token game in \cref{lemma:BK18residual,cor:joker-hd-buchi} yields the result of Bagnol and Kuperberg that 2-token game also characterises history-determinism on B\"uchi automata~\cite{BK18} as a corollary. This conclude the proof of \cref{cor:2token-game}.

\subsection{Solving Joker Games}
Joker games can be represented as a $[0,2]$-parity game, as we show below.
\begin{definition}[Joker game on a B\"uchi automaton]
For a  non-deterministic B\"uchi automaton $\Ac=(Q,\Sigma,q_0,\Delta)$, define the game $G_J (\Ac)=(V,E)$ with the parity condition $\Omega$ as follows. 
\begin{itemize}
    \item The set of vertices $V$ consists of the set $V = V_1 \cup V_2 \cup V_3$, where:
    \begin{enumerate}
        \item $V_1 = Q \times Q$
        \item $V_2 = Q \times \Sigma \times Q$
        \item $V_3 = Q \times Q \times \Sigma$
    \end{enumerate}
    Eve's vertices are $V_{\eve} = V_2$, while Adam's vertices are $V_{\adam} = V_1 \cup V_3$. The initial vertex is $(q_0,q_0)$.
    \item The set of edges $E$ is the union of following sets:
        \begin{enumerate}
            \item $E_1 = \{(p,q) \xrightarrow{} (p,a,q) \mid a \in \Sigma \}$. These edges correspond to Adam choosing a letter at the beginning of each round. 
            \item $E_2 = \{(p,a,q) \xrightarrow{} (p',q,a) \mid p\xrightarrow{a:c}p' \in \Delta\} $. These edges correspond to Eve choosing a transition on her token 
            \item $E_3 = \{(p',q,a) \xrightarrow{} (p',q') \mid q \xrightarrow{a:c}q' \in \Delta \}$. These edges correspond to Adam choosing a transition on his token, without playing a Joker
            \item $E_J=\{(p',q,a)\xrightarrow{} (p',q') \mid p \xrightarrow{a:c}q' \in \Delta   \}$. These edges correspond to Adam playing a Joker and choosing a transition from the state of Eve's token.
        \end{enumerate}
    \item The priority function $\Omega$ is defined as follows: All  elements in $E_1$ are assigned priority 0, while edges $(p,a,q) \xrightarrow{} (p',q,a)$ in $E_2$ are assigned priority 2 if the transition $\delta=p \xrightarrow{a:c} p'$ in $\Bc$ is accepting (or equivalently, $c=2$), and 0 otherwise. The edge $(p',q,a) \xrightarrow{} (p',q')$ in $E_3$ is assigned priority 1 if the transition $q\xrightarrow{a:c} q'$ is accepting, and 0 otherwise. The edges in $E_J$ are assigned priority 2.
\end{itemize}
\end{definition}
Note that the above game has $\Oc(|\Sigma| |Q|^2)$ many vertices and $\Oc(|\Sigma| |Q||\Delta|)$ many edges with priorities in $\{0,1,2\}$. Since such parity games with $n$ vertices and $m$ edges can be solved in time $\Oc(mn)$ time~\cite[Theorem 11]{Jur00}, we get that history-determinism of a B\"uchi automaton can be decided in time $\Oc(|\Sigma|^2|Q|^3|E|)$, an improvement from the $\Oc(|\Sigma|^2|Q|^4|E|^2)$ time of Bagnol and Kuperberg~\cite[Theorem 12]{BK18}.

\jokergamecomplexity*

%% file: 73appendix.tex
\subsection{Step-ahead Simulation}
We start by showing that the relations of \stepsimnoun and \racesimnoun is transitive. 
The transitivity of \stepsimnoun is easier to prove, as shown by the following lemma.
\begin{lemma}\label{lemma:dual-simulation-transitivity}
Let $\Ac$, $\Bc$ and $\Cc$ be non-deterministic parity automata such that $\Ac$ step-ahead simulates $\Bc$ and $\Bc$ \stepsim $\Cc$. Then $\Ac$ \stepsim $\Cc$.
\end{lemma}
\begin{proof}
Similar to \cref{lemma:g1todelaysim}, we can argue that the game arenas of the expanded game arenas of \stepsimnoun game between $\Ac$ and $\Bc$, and of simulation game between $\Ac$ and $\delay(\Bc)$ are equivalent. Thus, $\Ac$ simulates $\delay(\Bc)$ and $\Bc$ simulates $\delay(\Cc)$. Moreover, it is clear that $\delay(\Bc)$ simulates $\Bc$: in the simulation game between $\delay(\Bc)$ and $\Bc$, Eve in $\delay(\Bc)$ in $(i+1)^{th}$ round can simply `copy' Adam's transition on $\Bc$ in the $i^{th}$ round for each $i \geq 0$. All in all, we have that $\Ac$ simulates $\delay(\Bc)$, $\delay(\Bc)$ simulates $\Bc$, and $\Bc$ simulates $\delay(\Cc)$. By transitivity of simulation, we get that $\Ac$ simulates $\delay(\Cc)$, or equivalently, $\Ac$ \stepsim $\Cc$.
\end{proof}

Let us recall the definition of \racesimnoun for B\"uchi automata.
\cldeflocaldualsimulation*

We shall reduce \racesimnoun between two B\"uchi automata to \stepsimnoun between two \emph{reachability automata}. These are restricted B\"uchi automata where all the accepting transitions occur as self-loops on an accepting sink state, and this sink state has accepting self-loops on all letters in its alphabet. 

For each non-deterministic B\"uchi automaton $\Ac = (Q,\Sigma,q_0,\Delta_A)$, we construct the reachability automaton $\Ac_{\#}=(Q_\#,\Sigma_{\#},q_0,\Delta_\#)$ as follows. The automaton $\Ac_{\#}$ is over the alphabet $\Sigma_{\#}=\Sigma \cup \{\#\}$ for some $\# \notin \Sigma$. The set $Q_{\#}$ consists of the states $Q$ and two additional sink states $f$ and $r$. The transitions $\Delta_\#$ consists of all rejecting transitions in $\Delta$, and the following additional transitions.
\begin{enumerate}
    \item Transitions $p \xrightarrow{a:1} f$ for all accepting transition $p \xrightarrow{a:2} q$ in $\Ac$.
    \item Transitions $p\xrightarrow{\#:1} r$ for each state $p$ in $Q$.
    \item Transitions $f \xrightarrow{a:2} f$ and $r \xrightarrow{a:1}r$ for each $a \in \Sigma$.
\end{enumerate}

The following lemma relates \racesimnoun between two B\"uchi automata $\Ac$ and $\Bc$ to \stepsimnoun between $\Ac_{\#}$ and $\Bc_{\#}$.
\begin{lemma}\label{lemma:local-to-dual-sim}
Given two B\"uchi automata $\Ac$ and $\Bc$, $\Ac$ \racesim $\Bc$ if and only if $\Ac_\#$ \stepsim $\Bc_\#$.
\end{lemma}
\begin{proof}
    $\Rightarrow.$ Observe that Eve taking an accepting transition on her token in no later round than Adam in the \racesimnoun game between $\Ac$ and $\Bc$ corresponds to Eve's token reaching the state $f$ in no later round than Adam's in the \stepsimnoun game between $\Ac_{\#}$ and $\Bc_{\#}$. 

    $\Leftarrow.$ Suppose Adam has a strategy to ensure that his token can take an accepting transition on his token before Eve in the \racesimnoun game between $\Ac$ and $\Bc$. Then, Adam, by playing according to the same strategy in the \stepsimnoun game between $\Ac_{\#}$ and $\Bc_{\#}$, can ensure his token reaches $f$ before Eve's token does. From here, he can select $\#$ as letters for all the future rounds, causing Eve's run on her token to be rejecting while his run on his token is accepting.
\end{proof}
We utilise the above construction to prove \cref{lemma:dual-simulation-2} as well.
\cllemmadualsimulationtwo*
\begin{proof}
    It is clear that if $q$ is \lex, then $q$ \racesim itself. Thus, let $q$ be a state such that $(\Hc,q)$ \stepsim $(\Hc,q)$.   From \cref{lemma:local-to-dual-sim}, this is equivalent to $(\Hc_\#,q)$ step-ahead simulating itself, which in turn is equivalent to the Eve winning the 1-token game on $(\Hc_\#,q)$. Boker and Lehtinen have shown that $1$-token games characterises history-determinism on reachability automata and HD reachability automata have positional strategies in the history-determinism game which can be found in polynomial time via $1$-token games~\cite[Theorems 4.8 and 4.10]{BL23quantitative}.
    
    Fix an uniform positional strategy in the HD game from all states $q$ such that $(\Hc_{\#},q)$ is HD, and let $\Fc_{\#}$ be the deterministic subautomaton consisting of the transitions in this strategy.  
    Let $\Fc'$ be the largest common subautomaton of $\Hc$ and $\Fc_{\#}$. Observe that for each state $q\in Q$ and letter $a \in \Sigma$ such that there is a transition $q \xrightarrow{a:2} f$ in $\Hc_{\#}$, there is a transition $q\xrightarrow{a:2}p$ in $\Hc$ by construction. Adding one such transition in $\Fc'$ for each such state and letter pair, we get a deterministic subautomaton $\Fc$ of $\Hc$ that is a \lexnoun witness for all states $q$ that \racesim itself, as desired. 
\end{proof}

\paragraph*{Formal description of $\Dc$}
Let $\Hc$ be a history-deterministic automaton such that it is semantically-deterministic, and for every state $p$ in $\Hc$, there is a language equivalent state $q$ such that $p$ \racesim $q$. Fix the subautomaton $\Fc$ from \cref{lemma:dual-simulation-2}, and a positional Eve strategy $\tau: Q \times \Sigma \times Q \xrightarrow{} \Delta$ for all pairs of states $(p,q)$ such that $p$ \racesim $q$. Then, the deterministic automaton $\Dc$ is constructed as follows.
\begin{itemize}
    \item The states of $\Dc$ consists of states $(p,q)$ that are pairs of language-euivalent states in $\Hc$ such that $p$ \racesim $q$ and $q$ is \lex.  We let the initial state $d_0$ to be $(q_0,p_0)$ such that $p_0$ is \lex  and $q_0$ \racesim $r$. 
    \item Fix an ordering on the states of $\Hc$. $\Dc$ has the following transition on reading the letter $a$ from a state $(q,p)$. 
    \begin{itemize}
      \item The transition $\delta(q,p) \xrightarrow{a:2} (q',p')$, if there is an accepting transition $q\xrightarrow{a:2}q'$ in $\Ac$, where $p'$ is the least state in the ordering such that $p'$ is equivalent to $q'$,  $q'$ \racesim $p'$ and $p'$ is $\lex$.
        \item Otherwise, the transition $(q,p) \xrightarrow{a:1} (q',p')$, where $\tau(q,a,p)= (q \xrightarrow{a:1}q')$ is the transition prescribed by Eve's \racesimnoun strategy $\tau$,  and $p\xrightarrow{a:1} p'$ is the unique transition from $p$ on $a$ in $\Fc$. 
    \end{itemize}
\end{itemize}

\subsection{Towards Automata with Stepahead Simulation}
Recall that we call a history-deterministic B\"uchi automaton $\Hc$ good if $\Hc$ is semantically-deterministic and Eve wins the Joker game from all states in $\Hc$. We show that every history-deterministic B\"uchi automaton can be converted to an equivalent good HD B\"uchi automaton in polynomial time.

\begin{lemma}\label{lemma:good-automata}
Every history-deterministic B\"uchi automaton $\Hc$ has an equivalent good history-deterministic B\"uchi automaton $\Hc_N$ as a subautomaton which can be computed in polynomial time.
\end{lemma}
\begin{proof}
    Fix a positional winning strategy $\tau$ for Eve on the Joker game of $\Hc$. Consider the automaton $\Hc_N$ only consisting of transitions that can be taken via Eve's token in any play of the Joker game where Eve is playing according to $\tau$. We shall show that $\Hc_N$ is good. 
    
    Since $\Hc_N$ consists of transitions that are used by a winning strategy in the Joker game of $\Hc$ for Eve, it follows that Eve wins the Joker game on $\Hc_N$. Furthermore, we claim that Eve wins the Joker game from all states in $\Hc_N$. Suppose to the contrary that there is a state $q$ in $\Hc_N$ such that Eve loses the Joker game on $(\Hc_N,q)$. Then, consider the play of the Joker game where Eve is playing according to $\tau$ and Eve's token reaches the state $q$, Adam then plays Joker in the next round and plays to win the Joker game. Then, Adam wins the Joker game, which is a contradiction since $\tau$ is a winning strategy. We argue that the automaton $\Hc_N$ only consists of language-preserving transitions similarly.

    The automaton  $\Hc_N$ is also equivalent to $\Hc$. Indeed, for any word $w$ accepted in $\Hc$ with an accepting run $\rho$, consider the play of the Joker game where Eve plays according to $\tau$ while Adam moves his token according to $\rho$. Then, since $\tau$ is a winning strategy, Eve's run in her token is accepting. As the run of her token is entirely contained in $\Hc_N$, we get that $w \in \Lc(\Hc_N)$. 

    Since a winning strategy for Eve in the Joker game on $\Hc$ can be found in polynomial time, such an $\Hc_N$ can be constructed from $\Hc$ in polynomial time.
\end{proof}

Let us fix a good history-deterministic B\"uchi automaton $\Hc$. We recall the iterative procedure presented in \cref{subsec:normalising-automata}. We start by setting $\Hc_0 = \Hc$, and then do the following three steps on $\Hc_i$ for each $i$ till we get $\Hc_{i+1}=\Hc_i$.

\begin{enumerate}
    \item For vertices $(p,q) \in G_1(\Hc_i)$, compute $\opt_i(p)$ to be the optimal-ranks of states in $G_1(\Hc_i)$.
    \item Remove each transition $q \xrightarrow{a:1}q'$ with $\opt_i(q)<\opt_i(q')$ in $\Hc_i$ to obtain $\Hc'_i$.
    \item Make each transition $q \xrightarrow{a:1} q'$ with $\opt_i(q)>\opt_i(q')$ in $\Hc'_i$ accepting, to obtain $\Hc_{i+1}$.
\end{enumerate}

We will show that steps 2 and 3 preserves the invariants that $\Hc'_i$ and $\Hc_{i+1}$ are good and language-equivalent to $\Hc_i$. 

\begin{lemma}\label{lemma:iteration-correctness-1}
    $\Hc'_i$ is good and $\Lc(\Hc'_i,q)=\Lc(\Hc_i,q)$ for each state $q$ in $\Hc_i$.
\end{lemma}
\begin{proof}
    Fix $\tau_i$ to be a rank-optimal Eve strategy in $G_1(\Hc_i)$. Let $q$ be a state in $\Hc_i$. We will first show that Eve wins the \stepsimnoun game between $(\Hc'_i,q)$ and $(\Hc_i,q)$, by giving a winning strategy $\sigma$ for Eve. The strategy $\sigma$ will require Eve to store in her memory a state $m$ in $\Hc_i$, which we initialise to be $m=q$. 
    
    Throughout the play we will preserve the following two invariants.
    \begin{enumerate}
        \item[I1] The state of Eve's token, Adam's token, and Eve's memory are all equivalent in $\Hc_i$.
        \item[I2] If Eve's token and Eve's memory are at the states $p_j$ and $m_j$ respectively, then $(p_j,m_j)$ is a vertex in $G_1(\Hc_i)$, i.e., $p_j$ and $m_j$ can be reached from the initial state $q_0$ upon reading the same word.
    \end{enumerate}
    
    Let us also fix a strategy $\pi$ for Eve in the \stepsimnoun game between $(\Hc'_i,m)$ and $(\Hc'_i,q)$, which we know exists since $m$ and $q$ are language-equivalent, and $(\Hc'_i,m)$ is HD since Eve wins the Joker game on $(\Hc'_i,m)$ (\cref{cor:joker-game}). 

    Eve then plays according to $\sigma$ as follows.  At a position $(p_j,q_j)$ where Eve's memory state is $m_j$, suppose Adam chooses the letter $a_j \in \Sigma$. Let $\delta_j=p_j\xrightarrow{a_j:c_j}p'_j$ be the transition prescribed by the rank-optimal strategy $\tau_i$. We will distinguish between the following two cases.
    \begin{enumerate}
        \item If $\delta_j$ is a transition in $\Hc'_i$. Then Eve picks the transition $p_j \xrightarrow{a_j:c_j}p'_j$ in her token. Eve updates her memory $m_{j+1}$ according to the transition prescribed by $\pi$ in the \stepsimnoun game between $(\Hc'_i,m_j)$ and $(\Hc'_i,q_j)$.
        \item If $\delta_j$ is not a transition in $\Hc'_i$, i.e., it has been removed from $\Hc_i$ to obtain $\Hc'_i$. Then, let $m'$ be the state so that $\opt_i(p_j)=\rank_i(p_j,m')$. We then let Eve pick the transition on $a_j$ prescribed by $\tau_i$ in the \stepsimnoun game between $(\Hc_i,p_j)$ and $(\Hc_i,m')$. As for the memory, we fix another \stepsimnoun strategy $\pi'$ between $(\Hc_i,m')$ and $(\Hc_i,q_j)$. Eve will then update her memory according to transition prescribed by $\pi'$ on the letter $a_j$.
    \end{enumerate}

    This concludes the description of the strategy $\sigma$ for Eve. To see that this is a winning strategy, note that whenever Eve's token takes a rejecting transition in a play from the \stepsimnoun game between $(\Hc'_i,p_j)$ and $(\Hc_i,q_j)$, the quantity $\rank_i(p_j,m_j)$ is non-increasing, and strictly decreases if Eve makes a move of the second kind. 
    
    If Adam constructs an accepting run on his token, then either Eve makes a move of the second kind eventually, or the memory token takes an accepting transition (since $\pi$ is a winning strategy). Both of these scenarios cause the quantity $\rank_i(p_j,m_j)$ to strictly decrease. Since $\rank(p_j.m_j)$ is finite and cannot go below 0, Eve's token will take an accepting transition on her token in $\Hc'_i$ eventually if Adam constructs an accepting run on his token. But then, we can repeat this argument and use induction to show that Eve's token will take an accepting transition infinitely many times. Thus, Eve wins the \stepsimnoun game between $(\Hc'_i,q)$ and $(\Hc_i,q)$.

    Hence, $\Lc(\Hc'_i,q) \supseteq \Lc(\Hc_i,q)$, and since $\Hc'_i$ is obtained from $\Hc_i$ by deleting transitions, $\Lc(\Hc'_i,q)=\Lc(\Hc_i,q)$. This implies that $\Hc'_i$ is also SD, since $\Hc_i$ is SD. Furthermore, since $(\Hc'_i,q)$ \stepsim $(\Hc_i,q)$, $(\Hc'_i,q)$ also \stepsim itself (recall that $\Hc'_i$ is a subautotmaton of $\Hc_i$), implying Eve wins the Joker game from every state in $\Hc'_i$. It follows that $\Hc'_i$ is a good history-deterministic B\"uchi automaton.
\end{proof}

Note that the second step in the procedure ensures that the quantity $\opt_i$ is non-increasing across rejecting transitions. This can be used to prove easily that $\Hc_{i+1}$ is good and language-equivalent to $\Hc'_i$.

\begin{lemma}\label{lemma:iteration-correctness-2}
    $\Hc_{i+1}$ is good and $\Lc(\Hc_{i+1},q)=\Lc(\Hc'_i,q)$ for each state $q$ in $\Hc'_i$.
\end{lemma}
\begin{proof}
    It suffices to show that a run is accepting in $\Hc_{i+1}$ if and only if the same run is accepting in $\Hc'_i$, since then one can easily construct a winning strategy for Eve in the Joker game on $\Hc_{i+1}$ from her winning strategy in the Joker game on $\Hc'_i$.

    One direction is clear: if a run in $\Hc'_i$ is accepting, then the same run must be accepting in $\Hc_{i+1}$, since we only make certain rejecting transitions accepting.

    For the other direction, suppose that $\rho$ is an accepting run in $\Hc_{i+1}$, and consider the same run $\rho$ in $\Hc'_{i}$. Note that the $\opt_i$ values are non-increasing across rejecting transitions, and for every transition $\delta=q \xrightarrow{a:c}q'$ such that $\delta$ is an accepting transition in $\Hc_{i+1}$, the $\opt_i$ value strictly decreases. Thus, if $\rho$ has infinitely many accepting transitions in $\Hc_{i+1}$, then  $\rho$ must also have infinitely many accepting transitions in $\Hc'_i$ since $\opt_i$ is bounded by the size of the arena $G_1(\Hc'_i)$. 
\end{proof}

From \cref{lemma:iteration-correctness-1,lemma:iteration-correctness-2}, we conclude that our iterative procedure maintains the history-determinism, language and the property of the automaton being good through its steps.